\documentclass[runningheads]{llncs}

\usepackage{geometry}
\usepackage{framed}

\usepackage{amsmath,amssymb}
\usepackage{nicefrac}

\usepackage[noend,noline]{algorithm2e}
\RestyleAlgo{ruled}

\usepackage{tikz}
\usetikzlibrary{calc,arrows.meta,patterns}

\usepackage{hyperref}
\usepackage{cleveref}

\usepackage{xspace}
\newcommand{\bigO}{\ensuremath{O}}
\newcommand{\E}{\ensuremath{\mathbb{E}}\xspace}

\newcommand{\poly}{\text{poly}}

\title{Supervised Distributed Computing\thanks{John Augustine is supported by the Centre for Cybersecurity, Trust and Reliability (CyStar), IIT Madras. Christian Scheideler and Julian Werthmann are funded by the Deutsche Forschungsgemeinschaft (DFG, German
Research Foundation) – 549499840.}}

\author{John Augustine\inst{1}\orcidID{0000-0003-0948-3961} \and Christian Scheideler\inst{2}\orcidID{0000-0002-5278-528X} \and Julian Werthmann\inst{3}\orcidID{0000-0002-5110-5625}}

\authorrunning{J. Augustine, C. Scheideler and J. Werthmann}

\institute{
    IIT Madras, India \email{augustine@iitm.ac.in}
    \and Paderborn University, Germany \email{scheideler@upb.de}
    \and Paderborn University, Germany \email{jwerth@mail.upb.de}
}

\begin{document}
    
\maketitle
    
\begin{abstract}

We introduce a new framework for distributed computing that extends and refines the standard master-worker approach of scheduling multi-threaded computations. In this framework, there are different roles: a supervisor, a source, a target, and a collection of workers. Initially, the source stores some instance $I$ of a computational problem, and at the end, the target is supposed to store a correct solution $S(I)$ for that instance. We assume that the computation required for $S(I)$ can be modeled as a directed acyclic graph $G=(V,E)$, where $V$ is a set of tasks and $(v,w) \in E$ if and only if task $w$ needs information from task $v$ in order to be executed. Given $G$, the role of the supervisor is to schedule the execution of the tasks in $G$ by assigning them to the workers. If all workers are honest, information can be exchanged between the workers, and the workers have access to the source and target, the supervisor only needs to know $G$ to successfully schedule the computations. I.e., the supervisor does not have to handle any data itself like in standard master-worker approaches, which has the tremendous benefit that tasks can be run massively in parallel in large distributed environments without the supervisor becoming a bottleneck. But what if a constant fraction of the workers is adversarial? Interestingly, we show that under certain assumptions a data-agnostic scheduling approach would even work in an adversarial setting without (asymptotically) increasing the work required for communication and computations. We demonstrate the validity of these assumptions by presenting concrete solutions for supervised matrix multiplication and sorting.

\end{abstract}

\section{Introduction}

There is a long line of work in the literature on scheduling multithreaded computations (see, e.g., \cite{Leung2004} for a survey). Many approaches follow the master-worker paradigm, where a master delivers tasks among workers and collects their outcomes. In the most simple setting, there is a single master that delivers all work, but this limits scalability. However, with approaches such as MapReduce, where every worker can become a master as well so that work can be delivered recursively, the scalability can be significantly increased. In massive computations using a large number of workers, it cannot be assumed that all workers work reliably all the time. Therefore, suitable mechanisms are needed to obtain a desired degree of reliability. In the Hadoop MapReduce framework, for example, a worker is expected to report back periodically to its master with completed work and status updates. If a worker falls silent for longer than a given interval, the master records the worker as dead and sends the work assigned to that worker to another worker \cite{Hadoop}. While in a closed or trusted environment, it is reasonable to assume that the worst case that can happen to workers are crash failures, in situations where work leaves a trusted environment, one also has to worry about adversarial behavior. A popular approach where such behavior has to be taken into account is volunteer-based computing. An early example of this computing approach is SETI@home (1998-2020), where tasks were given to volunteers to analyze astronomical data for extraterrestrial life. To prevent cheating, SETI@home gave each task to two volunteers. If their results disagree, additional volunteers would be contacted for that task until enough results agree. Its follow-up service, BOINC, is nowadays used for a wide range of scientific computations. SETI@home and BOINC \cite{BOINC} both follow the approach that all tasks are delivered by a trusted server since it appears to be challenging to apply the MapReduce paradigm in a setting with untrusted volunteers. This server has to receive all input data from some source, which can be problematic for applications with massive amounts of data, like in data science and natural sciences, since that creates a bottleneck. If, instead, the server could just focus on scheduling computations without handling input and output data, it would be easy for a single server to schedule even millions of tasks. In fact, a recent study has shown that such an approach would indeed be beneficial \cite{Pastrana-CruzL23}. Thus, we will address the following central question: \\

\emph{Can the handling of data be decoupled from the problem of scheduling the execution of tasks, even under adversarial behavior of some of the workers?} \\

At first glance, that looks like a bad idea because how can a server know that tasks were executed correctly if it does not see the input or output. Alternatively, the source could do the checking, but the source might not have or be willing to provide the resources for output verifications. Thus, we are left with the workers. If there is an output verification mechanism for the tasks whose runtime is much lower than executing a task, an obvious strategy for the server would be the following: Ask a single worker to execute a task (based on input provided by the source) and send its output to a quorum of workers for verification. Suppose that the fraction of adversarial workers is sufficiently small and the quorum consists of a logarithmic number of randomly chosen workers. Then one can easily prove that if a majority of workers in that quorum tells the server that the output is correct, at least one honest worker in that quorum received the correct output, with high probability, irrespective of whether the executing worker is adversarial or not. However, for any constant fraction $\beta$ of adversarial workers with $\beta>0$, this approach has an inherent communication overhead of a logarithmic factor (given that the verification mechanism cannot easily be turned into a robust distributed version) because a logarithmic number of quorum members is necessary to be able to trust the quorum vote with high probability, and, in general, the quorum members need to know the entire input to be able to check whether the output is correct. Thus, the following question arises: \\

\emph{Is there a data-agnostic scheduling approach with a constant factor overhead for the computation and the communication?} \\

Interestingly, for the case that computations are given as task graphs, we  present a general framework that can achieve a constant factor overhead given that a lightweight verification mechanism is available for the tasks. Moreover, we demonstrate the validity of that assumption by presenting solutions for two standard problems: sorting and matrix multiplication. Our tailored verification mechanisms have the following approach in common, which might be of independent interest:

Standard verification mechanisms assume that a verifier knows the correct input and receives the output together with a certificate that allows it to check the correctness of the output. An exception are probabilistically checkable proofs, where a verifier can be convinced about the correctness of an output by just seeing a digest of the certificate. We are following a related approach in a sense that the server obtains digests of the outputs of the tasks that will help the workers assigned to a task to determine the correctness of the inputs received from preceding tasks. The problem with that approach is that the server does not know the instance stored in the source. But it can instead fetch a digest of the instance from the source, and we assume that the source is honest, so that the server has a trusted base for the local verifications.

Of course, it is not yet clear to which extent this approach can also be used for other problems, but in recent years, enormous progress has been made on light-weight verification, as discussed in the related works section below. Thus, we believe that our approach can also be applied to a wide range of other problems.

\subsection{Supervised distributed computing framework}

In our supervised framework, we have a \emph{supervisor}, a \emph{source}, a \emph{target}, and a collection of \emph{workers}. We use the word ``supervisor" instead of ``server" to emphasize our light-weight scheduling approach. The source just represents a storage environment with read-only access, while the supervisor and the target can execute algorithms. All three are assumed to be reliable in a sense that they operate on time and do not experience crashes or adversarial behavior. Given an instance $I$ of a computational problem, $I$ is initially stored in the source, and the goal of the supervisor is to use the workers to compute a solution $S(I)$ for $I$ that is ultimately stored in the target. Note that depending on the application, some of these roles might be associated with the same entity. For example, if $S(I)$ is small ("yes" or "no"), it might be convenient that the supervisor and the target are the same entity. On the other hand, the source might not be a single entity but multiple entities (e.g., webpages).

We will focus on approaches where the computation for $I$ can be represented as a directed acyclic graph (DAG) $G=(V,E)$, where each node $v \in V$ represents a task and $(v,w) \in E$ if and only if task $w$ needs information from task $v$ to be executed. We call such a graph a \emph{task graph}. The \emph{initial tasks} of $G$, which are tasks without incoming edges, need information from the source about $I$, while the \emph{final tasks} of $G$, which are tasks without outgoing edges, are supposed to send information to the target so that the target can assemble a solution for $I$. $G$ is assumed to be known to the supervisor. For most problems in this paper, we use $n=|V|$ to denote the size of $G$ and $D$ to denote the \emph{span} of $G$, which is the longest length of a directed path in $G$. $G$, $n$, and $D$ might depend on the size of $I$.

In this paper, we assume that the supervisor has access to a black-box worker-sampling mechanism that returns a worker that is adversarial with probability at most $\beta$ for some fixed $\beta \ge 0$, where $\beta$ is not known to the supervisor. Other ways of selecting workers can certainly be considered and are subject to future research. The number of workers might change over time. However, when an honest worker is selected for some task, we assume that it remains available and honest as long as the supervisor needs it for the computation. We allow all adversarial workers to be controlled by a single adversarial entity that can make decisions based on any information currently available in the system, but it does not know the random choices of the supervisor, the target, or the honest workers in the future. Thus, the adversary is omniscient w.r.t. the past and present but oblivious to the sampling process since it cannot convert an honest worker into an adversarial one \emph{after} being picked by the sampling process. Furthermore, the adversary cannot change, drop, delay, or reroute messages between the supervisor, source, target, and the honest workers.

For simplicity, we assume that time proceeds in synchronous \emph{rounds}. A round is long enough so that at the beginning of a round, the supervisor can select workers for all currently executable tasks and introduce them to all workers from which they need information (since a worker assigned to task $w$ needs the outputs of all tasks $v$ with $(v,w)\in E$ to execute $w$). On top of that, a round offers enough time for the selected workers to receive all required information (as long as both sides are honest) and perform the necessary verifications and computations for the tasks before the end of that round. This is reasonable if all tasks require approximately the same computational effort so that the overall computation can progress in an efficient way.

To find an efficient and robust supervised solution for a given problem $P$, several issues have to be addressed: First of all, instance $I$ might have to be stored in the source in a preprocessed way, and the time needed to preprocess $I$ should be bounded by $O(|I|)$. Second, a family of task graphs with low span, low total work, and low maximum work per task is needed for the given computational problem to allow a low runtime and high work efficiency and parallelism for the workers. And third, a robust scheduling strategy is needed for the supervisor to assign the tasks of the given task graph to the workers so that at the end a correct solution can be assembled at the target. The primary challenge here is to identify lightweight verification mechanisms for the supervisor as well as the workers that allow them to identify malicious behavior. Ideally, we would like to match the following bounds:
\begin{itemize}
\item Total preprocessing time and communication work of the source: $O(|I|)$.
\item Total work of supervisor: $\tilde{O}(|V|+|E|)$, where $\tilde{O}$ hides a polylogarithmic factor.
\item Total work of workers, including verifications: $O(W)$, where $W$ is the total work over all tasks in $G$, excluding verfications.
\item Runtime of the computation: $O(D)$.
\item Total work of the target: $O(S)$, where $S$ is the worst-case size of a solution for $I$.
\end{itemize}
By \emph{work} we mean here communication as well as computational work. This rules out trivial scheduling strategies like giving the entire computation to a single worker: Since a worker can only finish one task per round, the runtime would jump up to $O(|V|)$ in that case, which can be significantly larger than $D$. It also rules out getting the supervisor involved in inputs and outputs because then its work would be in the same order as the source or target, which can be significantly larger than the size of $G$.

\subsection{Use cases}

To motivate our framework, we discuss several use cases.

With our approach, services like SETI@home and BOINC can now be realized in a much more lightweight fashion for the server (w.r.t. storage as well as communication) since data handling would be decoupled from scheduling, without putting a much higher burden on the source. One can also envision to implement the supervisor in a standard cloud because with a small overhead, also the fees of using a cloud service would be very small.

Another use case would be \emph{purely} P2P-based volunteer computing. A standard approach for robust and scalable computations in P2P systems without mutual trust relationships is to partition the peers into random quorums of logarithmic size so that, w.h.p., the fraction of adversarial peers within a quorum is approximately equal to the total fraction of adversarial peers. Given that the fraction of adversarial peers is sufficiently low, the quorums can take care of the role of the supervisor. More precisely, whenever a quorum is supposed to execute some task, it selects one of its members at random to be the worker for that task, introduces that worker to the workers of the preceding tasks in the task graph by contacting their respective quorums, and then run a standard consensus algorithm to arrive at an agreement on the answer of the worker. Since the supervisor is data-agnostic, \emph{the quorum does not have to replicate any input or output data for this approach to work} so that the overhead of using a quorum is negligible given that the amount of data that needs to be handled is sufficiently high. As a by-product, this significantly improves the computational overhead of all \emph{provably} robust distributed computing solutions that have been proposed for P2P systems in the literature before, where the data handling and/or the computations have at least a logarithmic overhead. See the related work for further details.

Another interesting use case would be blockchain-based volunteer computing. In blockchains, space is a particularly delicate resource since blocks are of limited size and placing information in blocks is not for free. Thus, it is important to save as much space as possible to coordinate distributed computations via a blockchain. If the blockchain is used to emulate the supervisor, it just needs to store all information that a supervisor would have to handle, which only depends on the size of the task graph. Workers may then be assigned to tasks by announcing their interest in an executable task in the blockchain, so that the source or the predecessors in the task graph know which worker to send the information to. 

Working out the details might open up interesting new research directions. 

\subsection{Related work}

Several works proposed models that are related to our approach. A highly influential one is due to Blum et al.~\cite{BEGKN91}. They consider a setting where a data structure is under an adversary's control and operations on that data structure are initiated by a reliable checker. The goal is to come up with a strategy for the checker to verify the execution of the operations on that data structure using as little memory as possible so that any error in the execution of an operation will be detected by the checker with high probability. Their work inspired a large collection of works on authenticated data structures (see \cite{Tamassia03} for a survey). A standard assumption for that approach is that there is a trusted source, an untrusted server, and a client. The source provides the data for the server and a digest on the current state of the data to the clients. The server is supposed to organize the data received from the source in a desired data structure. Whenever the client asks the server to execute some request on that data structure, it requires the server to send it an answer together with a certificate that allows it to verify together with the digest from the source that the answer from the server is correct. Thus, the source role is similar to our source role, the untrusted server can be seen as the untrusted computation by the workers, and the client might be identified with the target. While research in this area has focused on data structures, some of the techniques, like using digests of the input, are also useful in our context.

Goodrich \cite{Goodrich08} adapted parallel fault diagnosis results algorithms to prevent participants in grid computing from cheating. Specifically, the supervisor reassigns previously computed tasks to other participants with a specific assignment strategy to identify false results. However, the supervisor still needs to receive all outputs resulting in an I/O bottleneck.

Closer to our approach are \emph{certifying algorithms} \cite{AlkassarBMR14}, which compute, in addition to an output, a witness certifying that the output is correct. A checker for such a witness, which is usually much faster than the original algorithm, then checks whether the output is indeed correct. Certifying algorithms have been found for various concrete problems in the sequential domain (e.g. \cite{McConnellMNS11}), which we expect to be useful for our approach as well, but so far no generic approach has been proposed in that area.

The problem of finding generic approaches for delegating work to potentially untrusted computational entities has been heavily studied in the cryptographic community (see, e.g., \cite{DBLP:conf/dlt/CrescenzoKKS22} for a survey). In this research area, a computationally weak client provides one (or more) computationally powerful servers with a program and an input. The server returns the output of the program on the input and some proof that the output is actually correct. This setting is formalized with interactive proof systems (see, e.g., \cite{feigenbaum1992overview} for an overview), where a prover aims to convince a verifier that some statement is true. For the delegation of computation to be useful, it is typically required that generating the proof takes no longer than executing the program, and verifying the proof takes time linear in the size of the input (both up to polylogarithmic factors). Especially but not only in the context of blockchains, non-interactive arguments have been considered, in the sense that the prover and the verifier do not need to exchange information in multiple rounds. Rather, the prover just provides one message (say, by posting it on the blockchain) that is sufficient to convince the verifier. An \emph{argument} is a relaxation of a proof. While the generation of proofs for false statements is normally required to be impossible, it suffices to be computationally infeasible for arguments. Notable classes of non-interactive arguments are SNARGs and SNARKs (see, e.g., \cite{chiesa2014succinct,nitulescu2020zk} for introductions to them).

Robust leader-based computations have been heavily pursued in the BFT (Byzantine fault tolerance), and blockchain communities (see, e.g., \cite{ZhangPMTDMZJ24} for a recent survey on BFT consensus protocols and blockchains) since in practice, they scale much better than decentralized approaches. There, leaders are either elected from the participants (a classical example is PAXOS \cite{PeaseSL80}) or are assumed to be trusted dealers so that, for example, a BFT algorithm can be initialized in an appropriate way (see, e.g., \cite{Kokoris-KogiasM20} for some background) or reliable broadcasting can be performed. The solutions proposed in these communities can certainly be used to perform robust distributed computations, but we are not aware of any works on how to use a trusted dealer to just \emph{supervise} (in our sense) executions of tasks for computational problems instead of getting directly involved in the data handling. A significant amount of work has also been invested in decentralized approaches, such as secure multi-party computations (see, e.g., \cite{ZHAO2019357} for a survey). However, their overhead under adversarial behavior (compared to the naive approach) is at least linear in the group size for all solutions presented so far, so that they are not helpful in our context to reduce the overhead to a constant.

Other directions assume the existence of a reliable cloud. Friedman et al.~\cite{FriedmanKK13} propose a model where processes exchange information by passing messages or by accessing a reliable and highly-available register hosted in the cloud. They show a lower bound on the number of register accesses in deterministic consensus protocols and provide a simple deterministic consensus protocol that meets this bound when the register is compare-and-swap (CAS). Other results on distributed computing with the cloud are considered in \cite{AfekGP21,AugustineBMPRT23}. In \cite{AfekGP21}, no adversarial behavior is considered, while in \cite{AugustineBMPRT23}, a certain fraction of the peers can be adversarial. The latter paper provides robust solutions for the Download problem, where the cloud stores a collection of bits and these bits must be downloaded to all of the (honest) peers, and the problem of computing the Disjunction and Parity of the bits in the cloud.

In algorithmic research on peer-to-peer systems, a common strategy to filter out adversarial behavior is to organize the peers into quorums of logarithmic size. Various strategies have been proposed to organize peers into random quorums (e.g., \cite{AwerbuchS09}), and various protocols have been presented that use the idea of quorums in order to ensure reliable operations in a peer-to-peer system (e.g., \cite{FiatSY05,HildrumK03,NaorW03}). There are also techniques that do not rely on groups (e.g., \cite{FiatS07,SaiaFGKS02}), but these still incur at least a logarithmic overhead compared to the setting where all peers are honest. Jaiyeola et al. \cite{JaiyeolaPSYZ18} investigate the case that quorums are only of $O(\log \log n)$ size and show that when combining this with a suitably chosen overlay network, all but a $O(\nicefrac{1}{\poly(\log n)})$ fraction of peers can successfully route messages to all but a $O(\nicefrac{1}{\poly(\log n)})$-fraction of the peers. As demonstrated by them, this can then be used to solve various important problems like routing or broadcasting with just $O(\poly(\log \log n))$ overhead per peer. However, they did not address the issue of reliable distributed computation.

\subsection{Our contributions}

In this paper, we consider a new approach for robust distributed computing, called supervised distributed computing. We first present a general framework that assumes, for simplicity, that the workers can locally and without the help of the supervisor verify the correctness of information received from other workers. For the case that the task graph is a path of length $n$ and the black-box sampling mechanism selects an adversarial worker with probability at most $\beta$ for some sufficiently small constant $\beta>0$, we obtain the following results (see Theorem~\ref{thm:path}):
\begin{itemize}
\item The number of rounds needed to successfully complete the computation is $(1+O(\beta))n$, w.h.p. (where ``w.h.p.'' stands for {\em with high probability}, meaning a probability of at least $1-1/n^c$ for some constant $c>1$ that can be chosen arbitrarily high);
\item on expectation, the source sends the input (to the initial task) $1+O(\beta)$ times and the target receives the solution (from the final task) 
$1+O(\beta)$ times; 
\item the expected total (computational and communication) work of the honest workers is within a $1+O(\beta)$ factor of optimal (i.e., when $\beta=0$).
\end{itemize}
All of these bounds are asymptotically optimal if the workers have to be selected via the black-box sampling mechanism, due to a simple observation: The expected number of times a worker has to be sampled for some task until some honest worker is chosen is
\[
 \sum_{i \ge 1} i \cdot \beta^{i-1} (1-\beta) = \frac{1}{1-\beta} = 1+\frac{\beta}{1-\beta} .
\]
Note that the bounds are not obvious since we are dealing with a mixed ad\-ver\-sa\-ri\-al-sto\-chas\-tic process, where adversarial workers are allowed to show \emph{any} behavior based on the supervisor strategy and the current state of the system. Consequently, the trivial solution of using the same worker for all tasks is not a good idea since an adversarial worker would simply behave honest until the last task, which would force the supervisor to redo the entire computation since it has no records of the task outputs. Thus, the best runtime guarantee that can be given in this case, w.h.p., would just be $O(n \log n)$ instead of $O(n)$.

For arbitrary DAGs of size $n$ with degree $d$ and span $D \ge \log n$ and a sufficiently small $\beta=O(1/d^{2+\epsilon})$ (for any constant $\epsilon>0$), we show that the number of rounds to successfully complete the computation is $O(D)$, w.h.p. Furthermore, any such DAG can be extended without increasing the runtime bound so that the expected communication work of the source and target are asymptotically optimal, as formally stated in Theorem~\ref{thm:arb}. The bound on $\beta$ is close to the best possible bound of $1/d$ when using our approach, as stated in Section~\ref{sec:dag}. For the analysis, we adapt a well-known technique for non-adversarial settings, the delay sequence argument (see, e.g., \cite{LeightonMRR94,Ranade91}), to a mixed adversarial-stochastic setting, which might be of independent interest. Note that we did not try to optimize the bounds on $\beta$ for the path and the DAG case since our primary focus is on showing the feasibility of our approach.

Afterwards, we present solutions for two concrete problems. 

For the matrix multiplication problem, we combine a well-known divide-and-conquer approach with Freivalds' algorithm for cross-checking matrix multiplications in order to obtain a task graph of span $O(\log n)$ for the multiplication of two $n \times n$-matrices. As long as $\beta>0$ is a sufficiently small constant, the number of rounds to complete the computation is $O(\log n)$, w.h.p., the expected overhead of the source and the target is a constant, and the expected computational work for the honest workers is $O(n^3)$ (see Theorem~\ref{thm:matrix}), which matches the sequential work for the standard matrix multiplication. Using more efficient approaches like Strassen's algorithm would result in better work bounds, but we did not explore these here.

For the sorting problem, we provide a new way of performing a parallel mergesort by designing a task graph representing a leveled network with $O(\log n)$ levels and $n$ nodes per level, which might be of independent interest. As summarized formally in Theorem~\ref{th:merge}, for any instance of $m$ numbers with $m \ge n \log n$ and a sufficiently low constant $\beta>0$, the number of rounds to complete the computation is $O(\log n)$, w.h.p., the expected overhead for the source and target is a constant, the expected total amount of computational work for the honest workers is $O(m \log m)$, and their expected total amount of communication work is $O(m \log n)$. Thus, our solution only incurs a constant overhead compared to the sequential mergesort algorithm.

The path case is considered in Section~\ref{sec:path} and the DAG case in Section~\ref{sec:dag}. The matrix multiplication solution is presented in Section~\ref{sec:matrix}, and the sorting solution is given in Section~\ref{sec:mergesort}. Due to space constraints, most of the proofs had to be dropped, but they can be found in the full version of the paper.

\section{The Case of Path Graphs} \label{sec:path}

Our starting point will be to solve the supervised distributed computing problem for the case that the task graph is a directed path of $n$ nodes $v_1,\ldots,v_n$, i.e., each node $v_i$ represents a task whose input is the output of $v_{i-1}$ (resp.~the input stored at the source if $i=1$) and the output of $v_n$ is the solution to the problem to be solved by that computation. 

We will exclusively focus on the scheduling part and therefore simply assume that an honest worker assigned to some task $v_i$, $1 \le i \le n$ is able to verify without the help of the supervisor whether the output of task $v_{i-1}$ is correct or not, and that the target can verify whether the output of task $v_n$ is correct. As mentioned in the model, we assume a round to be long enough so that it covers the time for receiving $v_i$'s input from $v_{i-1}$ (or the source if $i=1$), verifying it, executing $v_i$ and then sending a feedback on that to the supervisor, given that the workers assigned to $v_{i-1}$ and $v_i$ are honest. Also, recall that the supervisor has access to a black-box sampling mechanism that chooses an adversarial worker with probability at most $\beta$.

Consider the following scheduling approach. Initially, the supervisor $s$ picks a random worker $p_1$ for task $v_1$ and introduces it to the source so that it can obtain the given instance. Similar to the Hadoop MapReduce framework, $s$ expects $p_1$ to reply to it by the end of that round with a DONE, meaning that $p_1$ has finished the execution of task $v_1$. If $s$ does not receive that message by the end of the round, $s$ continues to pick a new random worker $p_1$ for task $v_1$ for the next round and introduces it to the source until it finally receives a DONE. Once this is the case, it will pick a random worker $p_2$ for $v_2$ and introduce $p_1$ and $p_2$ to each other (so that $p_1$ knows whom to send its output to and $p_2$ knows whom to expect the output from at the beginning of the next round).

In general, for any worker $p_i$ assigned to task $v_i$ at the beginning of some round, for some $i>1$, the supervisor introduces $p_{i-1}$ and $p_i$ to each other and expects $p_i$ to reply to it by the end of that round with one of the following messages:
\begin{itemize}
\item DONE: This means that $p_i$ is ready to forward its output to the worker assigned to $v_{i+1}$, resp. the target if $i=n$. 
\item REJECT: This means that $p_i$ rejects the input sent to it from $p_{i-1}$, or $p_i$ has not received anything from $p_{i-1}$ at the beginning of that round.
\end{itemize}
Note that for specific problems, more complex replies (like digests of outputs) might be expected by the supervisor, but we just use DONE as a placeholder here.

The supervisor $s$ will react to the reply of $p_i$ in the following way:
\begin{itemize}
\item If $s$ receives a DONE from $p_i$, $i<n$, it picks a random worker $p_{i+1}$ for $v_{i+1}$ in order to continue the execution at $v_{i+1}$. If $s$ receives a DONE from $p_n$, it asks $p_n$ to forward its output to the target. If $s$ receives a DONE from the target, the computation terminates.
\item If $s$ does not receive a reply from $p_i$, $s$ picks a new random worker $p'_i$ for $v_i$ in order to re-execute $v_i$.
\item If $s$ receives a REJECT from $p_i$ with $i>1$ or the target, it picks a new random worker $p_{i-1}$ for $v_{i-1}$ in order to continue the execution at $v_{i-1}$. We also call that a \emph{rollback}. For $i=1$ the supervisor selects a new worker for $v_1$ and introduces it to the source.
\end{itemize}
Given that the computation is currently at $v_i$, $s$ will remember the last worker $p_j$ assigned to $v_j$ for all $j<i$ so that if there is a rollback to $v_{i-1}$ in case of a REJECT from $p_i$, it will ask the worker last assigned to $v_{i-2}$ to deliver its output to the new worker assigned to $v_{i-1}$. If the rollback goes all the way back to $v_1$, the supervisor will choose a new worker for $v_1$.

\begin{theorem}
\label{thm:path}
For a path graph of size $n$ it holds: If $\beta \le 1/12$ then under any adversarial strategy, the supervised computation correctly terminates in $(1+O(\beta))n$ rounds, w.h.p. Furthermore, the source just needs to send the input $1+O(\beta)$ times and the target just needs to receive the solution $1+O(\beta)$ times, on expectation. Moreover, the total computational and communication work of the workers is within a $1+O(\beta)$ factor of optimal (i.e., when $\beta=0$).

\end{theorem}

We start the proof of this theorem with a bound on the runtime.

\begin{lemma} \label{th:line-time}
If $\beta \le \nicefrac{1}{5}$ and $t \ge (\nicefrac{4\beta}{1-4\beta}) n$ then for any adversarial strategy, the supervised computation terminates within $n + t$ rounds with probability at least $1-\nicefrac{1}{\poly(t)}$.
\end{lemma}
\begin{proof}
Surely, whenever some worker $p_i$ does not reply to $s$, $p_i$ must be adversarial, and whenever some $p_i$ sends a REJECT to the supervisor, either $p_i$ or $p_{i-1}$ must be adversarial. In either case, a new random worker will be chosen for the involved adversarial worker. Thus, the selection of an adversarial worker will only prevent the execution from proceeding at most once, and whenever that happens, the execution will roll back by at most one step.

According to the procedure of the supervisor, the supervisor picks a new worker uniformly and independently at random in each round. Let the binary random variable $X_i$ be 1 if and only if the worker picked in round $i$ is adversarial. Let $X=\sum_{i=1}^T X_i$, where $T=n+t$. If $X < t/2$ then the execution can roll back for a total of less than $t/2$ steps, which means that the execution will terminate within $T$ steps. Thus, it remains to bound the probability that $X \ge t/2$.

Suppose that $t/2 \ge 2\beta T$, which is true if $t \ge (\nicefrac{4\beta}{1-4\beta}) n$ and $\beta \le \nicefrac{1}{5}$. Since $\E[X_t] = \Pr[X_t=1] = \beta$,
\[
  \E[X] = \sum_{i=1}^T \E[X_i] = \beta T
\]
and since the $X_i$'s are independent, we can use Chernoff bounds with $\delta = (\nicefrac{t/2}{\beta T})-1$ to get
\[
  \Pr[X \ge t/2] = \Pr[X \ge (1+\delta) \E[X]] \le e^{-\nicefrac{\delta \E[X]}{3}}
  \le e^{-(\nicefrac{t}{6}-\nicefrac{t}{12})} = e^{-\nicefrac{t}{12}}.
\]
\end{proof}

A consequence of this result is the following lemma, which will be useful later in the paper.

\begin{lemma} \label{lem:line-majority}
Consider a path of length $\ell = c \ln n$ for a sufficiently large constant $c >0$. If $\beta < \nicefrac{1}{8}$ then for any adversarial strategy, the number of adversarial workers chosen during the computation is at most $(\nicefrac{2\beta}{1-4\beta}) \ell$, w.h.p. in $n$. In particular, since $\nicefrac{2\beta}{1-4\beta}<1/2$ for $\beta < \nicefrac{1}{8}$, this means that when the computation terminates, the majority of the workers assigned to the nodes on the path is honest.
\end{lemma}
\begin{proof}
If $\beta < \nicefrac{1}{8}$ then $\nicefrac{4\beta}{1-4\beta}<1$. If $\ell = c \ln n$ then it follows from the proof of the previous lemma that, for $t=(\nicefrac{4\beta}{1-4\beta}) \ell$, the probability that the computation does not terminate within $\ell+t$ rounds is at most
\[
  e^{-\nicefrac{t}{12}} = e^{-(\nicefrac{\beta}{3(1-4\beta)}) \ell}
  = n^{-c \cdot \nicefrac{\beta}{3(1-4\beta)}}
\]
This is polynomially small in $n$ if the constant $c$ is large enough. Since the time bound is based on the fact that the probability that at least $t/2$ adversarial workers are selected within $\ell+t$ rounds is polynomially small for the chosen $\ell$ and $t$, the majority statement holds.
\end{proof}

Next, we bound the expected I/O work of the source and the target.

\begin{lemma} \label{th:line-input}
If $\beta \le \nicefrac{1}{12}$ then under any adversarial strategy, the source just needs to send the input $1+O(\beta)$ times on expectation.
\end{lemma}
\begin{proof}
First, we need to determine the worst possible adversarial strategy for the number of recomputations at $v_1$ (and therefore the number of times the source needs to send its input). Let us model a computation as a sequence $(i_1,s_1), (i_2,s_2),\ldots \in \{1,\ldots,\ell\} \times \{A,H\}$ where $i_t$ represents the task $v_{i_t}$ that needs to be executed at round $t$, $s_t=A$ means that an adversarial worker was chosen at round $t$, and $s_t=H$ means that an honest worker was chosen at round $t$. A computation is \emph{valid} if $i_1=1$ and for every round $t\ge 1$ before the computation terminates (if it does so), either (1) $i_{t+1}=i_t+1$, or (2) $i_{t+1}=i_t$ and $s_t=A$ (i.e., nothing was sent to the supervisor at round $t$), or (3) $i_{t+1}=i_t-1$, and $s_t=A$ or the worker last chosen for $v_{i_t-1}$ in $(i_1,s_1), (i_2,s_2),\ldots, (i_{t-1},s_{t-1})$ is adversarial (i.e., a rollback occurred at round $t$). Given a valid computation $C$, let $A(C)$ be any adversarial strategy that conforms with $C$, i.e., in case (1) it causes the computation to proceed (if any adversary is involved), in case (2) it sends nothing back to the supervisor, and in case (3) it causes a rollback. Our goal will be to show that a worst-case adversarial strategy concerning the number of times the source has to send the input is that an adversarial worker always sends REJECT to the supervisor.

Consider any valid computation $C=(i_1,s_1),(i_2,s_2),\ldots$ with at least one round $t$ where $s_t=A$ and the adversarial worker does \emph{not} send REJECT to the supervisor. Let $r$ be the first such round. Since before round $r$ the adversarial workers always caused rollbacks, all workers currently assigned to positions less then $i_r$ must be honest, i.e., the path configuration at round $r$ is given as $P_r=(p_1,\ldots,p_{i_r-1},s_{i_r})$ where $p_k = H$ is the state of the worker currently assigned to node $v_k$, for all $k \in \{1,\ldots,i_r-1\}$, and $s_{i_r}=A$ is the state of the worker assigned to $v_{i_r}$ in round $r$. Let $A(C)$ be any adversarial strategy that conforms with $C$. Our goal will be to show that there is an adversarial strategy $A'$ based on the \emph{same} state sequence $s_1,s_2,\ldots$ as $C$ with the following properties:
\begin{itemize}
\item It is identical to $A(C)$ for the first $r-1$ rounds but chooses a rollback instead at round $r$, and
\item the computation $C(A')$ resulting from $A'$ \emph{dominates} $C$, in a sense that the number of times the source needs to send the input is at least as large as in $C$.
\end{itemize}
Using that argument inductively (i.e., we continue with a transformation of $A'$ like for $A$, and so on) will then prove that the worst-case adversarial strategy is indeed that a selected adversarial worker always sends REJECT to the supervisor.

If we deviate from $C$ by performing a rollback at round $r$, then we reach the configuration $P'_{r+1}=(p_1,\ldots,p_{i_r-2},s_{i_{r+1}})$ if $i_r>1$, and otherwise we force the source to send the input again and reach the configuration $P'_{r+1}=(s_{i_{r+1}})$. For the original computation $C$ there are two possibilities for the follow-up configuration $P_{r+1}$:

Case 1: $P_{r+1}=(p_1,\ldots,p_{i_r-1},s_{i_{r+1}})$, which is the case if the adversary does not send a message to the supervisor. If $i_r=1$ then $P_{r+1}=P'_{r+1}=(s_{i_{r+1}})$ while $A'$ was able to force the source to send the input again. If $A'$ behaves exactly like $A$ from that point on, the remaining computations of $C$ and $C'$ will be identical, which implies that $C'$ dominates $C$. If $i_r>1$ then the currently executed position in $P'_{r+1}$ is one step closer to $v_1$ than in $P_{r+1}$. Let $A'$ behave in exactly the way as $A$ from that point on. Then $C(A')$ remains one step closer to $v_1$ compared to $C$ till a round $r'$ is reached where $A'$ forces the source to send the input again, or none of them cause the source to send the input again. In the latter case, $C(A')$ clearly dominates $C$, and in the former case, $P_{r'+1}=P'_{r'+1}=(s_{i_{r'+1}})$, i.e., we are back to the same situation as for $i_r=1$, proving dominance also in this case.

Case 2: $P_{r+1}=(p_1,\ldots,p_{i_r},s_{i_{r+1}})$, which is the case if the adversary sent a DONE to the supervisor. Note that in this case $p_1,\ldots,p_{i_r-1}=H$ while $p_{i_r}=A$. We distinguish between two further cases:

Case 2a: If $i_r=1$ then $P'_{r+1}=(s_{i_{r+1}})$ while $A'$ was able to force the source to send the input again. If, in $C$, the adversarial worker at position $i_r$ behaves like an honest worker from that point on, we are back to case 1. Otherwise, the only critical case is where $P'_{r'}=(H)$ and $P_{r'}=(A,H)$ for some round $r'$ and the adversarial worker sends a wrong output to the honest worker in $C$, causing a rollback. Then we end up in a situation where $P'_{r'+1}=(H,s_{i_{r'+1}})$ and $P_{r+1}=(s_{i_{r'+1}})$, i.e., the roles of $C(A')$ and $C$ are exchanged compared to case 1. However, $A'$ simply continues to behave like $A$ from that point on until $A$ causes the source to send the input again. If this never happens, $C(A')$ clearly dominates $C$. Otherwise, we are back in a setting where $P_{r''}=P'_{r''}=(s_{i_{r''}})$ while $C(A')$ and $C$ caused the same number of input transmissions of the source. Hence, if $A'$ continues to be behave like $A$, $C(A')$ dominates $C$.

Case 2b: It remains to consider the case that $i_r>1$. Suppose that in $C$ the adversarial worker at position $i_r$ behaves like an honest worker from that point on. Then $A'$ simply behaves like $A$ from that point on, and in this way $C(A')$ always stays two positions ahead of $C$ till a round $r'$ is reached where $A'$ forces the source to resend the input. In that case, $P'_{r'+1}=(s_{r'+1})$ while $P_{r+1}=(H,s_{r'+1})$, i.e., we are back to case 1. If never a round $r'$ is reached where $A'$ forces the source to resend the input, $C(A')$ dominates $C$ as well. Similar to case 2a, the only critical situation that is left is that $P'_{r'}=(p_1,\ldots,p_{i_r-2},H)$ and $P_{r'}=(p_1,\ldots,p_{i_r-1},A,H)$ and the adversarial worker at position $i_r-1$ sends a wrong output to the honest worker in $C$, causing a rollback. Then we end up in a situation where $P'_{r'+1}=P_{r'+1}$, i.e., if $A'$ behaves from that point on like $A$, we again have the situation that $C(A')$ dominates $C$. 

Thus, the adversarial strategy of always using a rollback is dominating all other strategies w.r.t. the number of input transmissions of the source. Suppose that it takes $T\ge 1$ rounds from the time that $v_1$ was executed until the next time that $v_1$ is executed. Let the binary random variable $X_t$ be 1 if and only if an adversarial worker was selected at round $t$, and let $Y_T=\sum_{t=1}^T X_t$. For $T=1$, $Y_T=X_1=1$ and for every even $T>1$, $Y_T=\nicefrac{T}{2}$ so that indeed $v_1$ is executed again after $T$ rounds. The probability for $T=1$ is $\beta$. For any even $T>1$, there are at most ${T \choose \nicefrac{T}{2}} \le 2^T$ ways of choosing a distribution of $\nicefrac{T}{2}$ adversarial workers over $T$ steps. Thus, $\Pr[Y_T=\nicefrac{T}{2}] \le 2^T \beta^{\nicefrac{T}{2}}$, which implies that for $\beta \le \nicefrac{1}{12}$ the probability to re-execute $v_1$ is at most
\[
  \beta + \sum_{{\rm even} \; T>1} (4 \beta)^{\nicefrac{T}{2}} \le \beta + \sum_{T \ge 1} (4\beta)^T = \beta (1 + \nicefrac{4}{1-4\beta}) \le 7\beta.
\]
Therefore, the probability that $v_1$ needs to be executed at least $i+1$ times is at most $(7\beta)^i$, which implies that the expected number of times the source needs to send the input is at most
\[
  \sum_{i \ge 0} (7 \beta)^i = \frac{1}{1-7\beta} = 1+O(\beta).  
\]
\end{proof}

We can significantly strengthen this result, which will be useful when using paths as a substructure of some task graph:

\begin{lemma} \label{lem:advsupervisor}
Consider a path of length $\ell = c \log n$ for a sufficiently large constant $c \ge 1$ in which the target always rejects the output sent by $v_{\ell}$. If $\beta \le 1/12$, then under any adversarial strategy, the expected number of times the source sends the input within $n^{c-1}$ many rounds is $1+O(\beta)$.
\end{lemma}

\begin{proof}
Like in the proof of Lemma~\ref{th:line-input}, let us model a valid computation $C$ as a sequence $(i_1,s_1), (i_2,s_2),\ldots \in \{1,\ldots,\ell\} \times \{A,H\}$ where $i_t$ represents the task $v_{i_t}$ that needs to be executed at round $t$, $s_t=A$ means that an adversarial worker was chosen for round $t$, and $s_t=H$ means that an honest worker was chosen at round $t$. Since the target always rejects the output, that computation will have an infinite length. In the same way as in the proof of Lemma~\ref{th:line-input}, we can transform $C$ into a valid computation $C'$ that dominates $C$ w.r.t. the number of times the source has to send the input and where an adversarial worker always sends REJECT to the supervisor. We now want to prove an upper bound for the expected number of times the source sends the input if we stop $C'$ after $n^c$ many rounds.

First, of all, from the proof of Lemma~\ref{th:line-input} we know that the source just needs to send the input $1+O(\beta)$ times until we reach a point where an honest worker assigned to the last task of the path sends its output to the target for the first time. At that point, all workers associated with the tasks are honest in computation $C'$. In that case, it is very unlikely to roll back the computation to $v_1$, as we will see.

Suppose that it takes $T \ge \ell$ rounds to roll back the computation from $v_{\ell}$ to $v_1$ without again reaching the point that an honest worker assigned to the last task sends its output to the target. Then for at least $T/2+\ell/2$ of the $T$ rounds an adversarial worker must be chosen. For $\beta \le 1/12$, the probability for that is at most
\begin{align*}
    { T \choose T/2+\ell/2} \beta^{T/2+\ell/2} &\le \left( \frac{e T}{T/2+\ell/2} \right)^{T/2+\ell/2} \beta^{T/2+\ell/2}\\
    &\le (2e\beta)^{T/2+\ell/2}\\
    &\le (1/2)^{T/2+\ell/2}
\end{align*}
Thus, the probability that there exists such a $T$ is at most
\[ 
  \sum_{T \ge \ell} (1/2)^{T/2+\ell/2} \le 2 \cdot (1/2)^{\ell} = 2/n^c
\]

Therefore, within $n^{c-1}$ many rounds, the expected number of attempts to get from $v_\ell$ back to $v_1$ without sending an output to the target is bounded by $O(1/n)$, which implies that the expected number of times the source sends the input is still bounded by $1+O(\beta)$.
\end{proof}

Another consequence of Lemma~\ref{th:line-input} is the following result:

\begin{lemma} \label{th:line-middle}
If $\beta \le \nicefrac{1}{12}$ then it holds for any task $v_i$ and any adversarial strategy that the expected number of times an honest worker has to execute $v_i$ is $1+O(\beta)$.
\end{lemma}
\begin{proof}
Consider any time where the computation reaches the point that an honest worker executes task $v_i$, which means that it accepted the input from the worker responsible for $v_{i-1}$ (resp. the source). Then the computation is guaranteed to continue with $v_{i+1}$, and the worker for $v_i$ will correctly deliver its output to $v_{i+1}$. That is, we are exactly in the situation of Lemma~\ref{th:line-input} where the source sends the given input to $v_1$. Thus, the expected number of times an honest worker has to execute $v_i$ is at most the expected number of times the source has to send the given input ("at most" since some cases cannot occur any more because the remaining sequence of nodes from $v_i$ is shorter).
\end{proof}

Suppose that the (communication and computational) work per task is 1 and the verification work for the input from some task is at most $\epsilon <1$. In that case, Lemmas~\ref{th:line-time} and \ref{th:line-middle} immediately imply the following result, which is asymptotically optimal.

\begin{corollary}
If $\beta \le 1/12$ then for any adversarial strategy, the expected work the honest workers have to invest to execute the tasks of $G$ is $(1+O(\beta))(1 + \epsilon)n$.
\end{corollary}

It remains to bound the number of times the target has to receive the output.

\begin{lemma} \label{th:line-output}
If $\beta \le \nicefrac{1}{3}$ then under any adversarial strategy, the target just needs to receive the output from $v_{\ell}$ $1+O(\beta)$ times on expectation.
\end{lemma}
\begin{proof}
Certainly, if the adversary always either performs a rollback or does not respond whenever an adversarial worker is selected, the number of times the target receives the output is guaranteed to be at most 1, namely, at that point when all $v_i$'s are associated with honest workers.

Therefore, let us ignore the case for now that the adversary chooses to do a rollback or not to respond for now, which reduces its choices to either sending a correct or wrong output to the worker associated with the next task.

First, consider the simplest case that an adversarial worker was chosen for $v_\ell$. Then the optimal strategy is to send a wrong output to the target so that the supervisor chooses a new worker for $v_{\ell}$, which opens up the opportunity for another output transmission to the target. Next, consider the case that an adversarial worker $p_{\ell-1}$ was chosen for $v_{\ell-1}$. Since we do not consider the strategy for now that the adversary sends REJECT to the supervisor or does not respond, it would respond with DONE, which causes the supervisor to pick a worker $p_\ell$ for $v_\ell$. Certainly, whenever $p_\ell$ is adversarial (which $p_{\ell-1}$ can find out since $p_\ell$ is introduced to it), the optimal strategy for $p_{\ell-1}$ is to let $p_\ell$ send a wrong output to the target. If $p_\ell$ is honest, we distinguish between two strategies for $p_{\ell-1}$.

Strategy H: $p_{\ell-1}$ delivers the correct output to $p_\ell$. Then the execution will terminate after $p_\ell$ has sent its output to the target, which will end $p_{\ell-1}$'s term.

Strategy A: $p_{\ell-1}$ delivers the wrong output to $p_\ell$. Then $p_\ell$ will notice that and send a REJECT to the supervisor, causing the supervisor to replace $p_{\ell-1}$ by a new random worker, which will end $p_{\ell-1}$'s term as well.

First, let us compute the expected number of outputs to the target if $p_{\ell-1}$ goes for strategy H. Since the probability that $i$ outputs are sent to the target is $\beta^{i-1} (1-\beta)$, this is
\[
  h_{\ell-1} = \sum_{i \ge 0} (i+1) \beta^i (1-\beta) = \frac{1}{1-\beta}
\]
Note that strategy H is also what an honest worker $p_{\ell-1}$ would do. If the adversary instead chooses strategy A, the expected number of outputs to the target until $p_{\ell-1}$'s term ends (which happens if $p_{\ell}$ is honest) is
\[
  a_{\ell-1} = \sum_{i \ge 1} i \cdot \beta^i (1-\beta) = h_{\ell-1} -1 = \frac{\beta}{1-\beta}
\]
Clearly, the better strategy at that level is H, and a mix of A and H would result in a convex combination of the two expectations that would lie between these two bounds. However, strategy H leads to a termination while strategy A continues the computation, so we also have to look at nodes of lower index. Note that a no-response or REJECT strategy of the adversary would clearly be inferior to both H and A since it results in an expected number of outputs of 0 until $p_{\ell-1}$'s term ends.

Next, consider the case that an adversarial worker $p_{\ell-2}$ was selected for $v_{\ell-2}$. We again distinguish between the two strategies H and A that $p_{\ell-2}$ can go for.

Suppose that $p_{\ell-2}$ goes for H. If an adversarial worker at node $v_{\ell-1}$ would also go for H, then the expected number of outputs to the target is equal to $h_{\ell-1}$. If, instead, an adversarial worker at node $v_{\ell-1}$ would go for A, then the expected number of outputs to the target until $p_{\ell-2}$'s term ends is
\begin{eqnarray*}
  h_{\ell-2} & = & (1-\beta) h_{\ell-1} + \beta(1-\beta) (a_{\ell-1} + h_{\ell-1}) +   \beta^2 (1-\beta) (2 a_{\ell-1} + h_{\ell-1}) + \ldots \\
  & = & h_{\ell-1} (1-\beta) \sum_{i \ge 0} \beta^i + a_{\ell-1} \beta (1-\beta) \sum_{i \ge 0} (i+1) \beta^i \\
  & = & h_{\ell-1} + \frac{\beta}{1-\beta} a_{\ell-1}
\end{eqnarray*}
This is higher than the expectation if the adversarial worker at node $v_{\ell-1}$ goes for H and also better than any mixed strategy of an adversarial worker at node $v_{\ell-1}$, since that would result in an expectation between these two bounds. Note that that bound would also be the outcome if $p_{\ell-2}$ is honest.

Suppose now that $p_{\ell-2}$ goes for A. If an adversarial worker at node $v_{\ell-1}$ would go for H, then the expected number of outputs to the target until $p_{\ell-2}$'s term ends is equal to $\beta \cdot h_{\ell-1} = a_{\ell-1}$. If, instead, an adversarial worker at node $v_{\ell-1}$ would go for A, then the expected number of outputs to the target until $p_{\ell-2}$'s term ends is
\[
  a_{\ell-2} = \sum_{i \ge 1} \beta^i (1-\beta) \cdot i \cdot a_{\ell-1} = \frac{\beta}{1-\beta} a_{\ell-1}.
\]
Again, this is lower than for the case that the adversarial worker at node $v_{\ell-2}$ goes for H, but it would lead to a termination while that would not happen if the worker at $v_{\ell-2}$ goes for A. Also, note again that a no-response or REJECT strategy of the adversary would clearly be inferior to H as well as A since it results in an expected number of outputs of 0 until $p_{\ell-2}$'s term ends.

Continuing with these arguments upwards, it turns out that the highest expected number of outputs to the target given that an adversarial worker at node $v_i$ goes for H is
\[
  h_i = h_{i+1} + \frac{\beta}{1-\beta} a_{i+1}
\]
and the highest expected number of outputs to the target given that an adversarial worker at node $v_i$ goes for A is
\[
 a_i = \frac{\beta}{1-\beta} a_{i+1}.
\]
Since the target is always honest, the highest expected number of outputs to the target is $h_0 = h_1 + (\nicefrac{\beta}{1-\beta}) a_1$. In order to bound that, one can inductively show that $a_i \le [\nicefrac{\beta}{1-\beta}]^{\ell-i} h_i$, starting with $i=\ell-1$, which implies that $h_i \le (1+[\nicefrac{\beta}{1-\beta}]^{\ell-i+1}) h_{i+1}$. Therefore,
\begin{align*}
  h_0 &\le \prod_{i=2}^{\ell} (1+[\nicefrac{\beta}{1-\beta}]^i) h_{\ell-1}\\
  &\le \exp{\left( \sum_{i=2}^{\ell} [\nicefrac{\beta}{1-\beta}]^i \right)} \frac{1}{1-\beta}\\
  &= \left( 1+ O(\beta^2) \right) \frac{1}{1-\beta}\\
  &= 1+O(\beta)
\end{align*}
if $\beta \le \nicefrac{1}{3}$.
\end{proof}

Combining the above lemmas results in Theorem~\ref{thm:path}.

%However, the work for the supervisor should be low as well. The next two results indeed show that, on expectation, the supervisor just needs to send the input $O(1)$ times, and it also just needs to receive the output $O(1)$ times if $\beta>0$ is sufficiently small. On top of that, the supervisor just has to assign $O(\ell)$ tasks and perform $O(\ell)$ worker introductions on expectation. Thus, its expected work is asymptotically optimal.

\section{The Case of Arbitrary DAGs}\label{sec:dag}

Next, we consider an arbitrary DAG $G=(V,E)$ with span $D$. It is well-known that the nodes $v \in V$ can be labeled with numbers $\ell(v) \in \{0,\ldots,D\}$ so that the nodes are topologically sorted, i.e., for all $(v,w) \in E$, $\ell(v) < \ell(w)$. Suppose on the contrary that there is a node $v$ that requires a label larger than $D$ for the nodes to be topologically sorted. Start with the node $v$ of highest label. Certainly, $v$ must have a predecessor $u$ with $\ell(u)=\ell(v)-1$ since otherwise the label of $v$ can be reduced. Continuing with this argument would result in a directed path of length more than $D$, contradicting the assumption that $G$ has a span of $D$. If a labeling with labels in $\{0,\ldots,D\}$ can be found so that for all $(v,w) \in E$, $\ell(w)=\ell(v)+1$, then $G$ is called a \emph{leveled network}. If this is not the case for some $(v,w) \in E$, then we may simply replace $(v,w)$ with a path of $\ell(w)-\ell(v)$ edges, where the only purpose of the inner nodes is to verify the output of the predecessor and, if found correct, to forward it to the successor. This creates more work, but if the verification work is insignificant compared to the work of executing the tasks, this is justified for the following reason:

The path replacement strategy is necessary if we want to reach a runtime of $O(D+\log n)$ w.h.p. To see that, consider a task graph consisting of a path of tasks $v_0,\ldots,v_D$ and an edge from $v_0$ to $v_D$. Whenever an adversarial worker is selected for $v_0$, its optimal strategy is to play honest till the computation has reached $v_D$ and then to send a wrong input to $v_D$, which causes the entire computation to roll back to $v_0$. Certainly, in this case the best possible runtime bound that can be shown to hold w.h.p. for a constant $\beta>0$ is $O(D \log n)$.

For the rest of this section, we assume w.l.o.g. that $G$ is a leveled network. We again exclusively focus on the scheduling part and therefore assume for simplicity that an honest worker assigned to some task $v \in V$ is able to verify without the supervisor whether the outputs of the preceding tasks (i.e., the tasks $u \in V$ with $(u,v) \in E$) are correct or not. Furthermore, we again assume a round to be long enough so that an honest worker can handle all tasks assigned to it, i.e., for every such $v$, the worker can receive the outputs from all predecessors of $v$ (if these are honest) and verify these, execute $v$ and then send a feedback on that to the supervisor.

In order to schedule the computation in the leveled network, the supervisor maintains a set $F \subseteq V$ of \emph{finished} tasks, i.e., tasks whose executions it already considers to be done, and for each task $v \in F$ it also remembers the worker $p_v$ last assigned to it. 
$F$ has to satisfy the following invariant: For all tasks $u \in V$ that have a directed path to a task $v \in F$, $u \in F$ as well. 

Based on $F$, we define the \emph{wavefront} $W_F \subseteq V$ of executable tasks as the set of all $v \in V \setminus F$ whose predecessors are all in $F$. Initially, $F=\emptyset$, and therefore, $W_F$ is equal to the set of all initial nodes (i.e., nodes without incoming edges) in $G$. Our scheduling strategy for path graphs can easily be generalized to a leveled graph. In particular, given that worker $p_v$ has been assigned to some executable task $v$, the supervisor expects one of the following replies from $p_v$:
\begin{itemize}
\item DONE: This means that $p_v$ is ready to forward its output to its successors in $G$. If $v$ is a final node, $p_v$ forwards its output to the target, instead.
\item REJECT($R$): This means that $p_v$ rejects the outputs from a subset $R$ of the predecessors of $v$.
\end{itemize}
The supervisor $s$ reacts to the reply of $p_v$ in the following way:
\begin{itemize}
\item If $s$ receives a DONE from $p_v$, then it adds $v$ to $F$.
\item If $s$ does not receive a reply from $p_v$, then $s$ leaves $F$ as it is.
\item If $s$ receives a REJECT($R$) from $p_v$, then $s$ removes all $w \in R$ as well as all $w' \in F$ reachable by a task $w \in R$ from $F$ in order to maintain the invariant for $F$.
\end{itemize}
Once we reach a round where $F=V$, i.e., all final nodes sent their outputs to the target, and the target accepted them, the execution terminates. If the target rejects an output of some final node $v$, it will not add $v$ to $F$ (i.e., it treats this case like a missing reply).

Let $d$ be an upper bound on the indegree and outdegree of a node in $G$ and $n=|V|$. The rest of the section is dedicated to the proof of the following theorem:

\begin{theorem}
\label{thm:arb}
For any DAG $G$ of degree $d$ and span $D$, $G$ can be extended to a graph $G'$ so that it holds: If $\beta\leq(\nicefrac{1}{2(2d+1)})^{2+\epsilon}$ for any constant $\epsilon>0$ then under any adversarial strategy, the supervised computation correctly terminates in $O((D+\log n)/\epsilon)$ rounds, w.h.p. Furthermore, the source and the target just need to send the input and receive the output $O(1)$ times, on expectation.
\end{theorem}

We start with a bound on the runtime.

\begin{lemma} \label{th:dag-runtime}
If $\beta\leq(\nicefrac{1}{2(2d+1)})^{2+\epsilon}$ for any constant $\epsilon>0$ then for any adversarial strategy, the supervised computation terminates within $O((D+\log_d n)/\epsilon)$ rounds w.h.p.
\end{lemma}
\begin{proof}
We will do a backwards analysis similar to the delay sequence arguments used for routing in leveled networks (see, e.g., \cite{LeightonMRR94,Ranade91}). However, since these arguments did not take adversarial behavior into account, some important adaptations are needed to make it work in our case. Without loss of generality, we assume that the adversarial workers only send messages that are specified by the protocol, since all other messages will be ignored by the supervisor, the source, the target, and the honest workers. Then the adversary is left with the following choices whenever the supervisor asks an adversarial worker $p$ to execute task $v$:
\begin{itemize}
\item $p$ does not send a message to the supervisor.
\item $p$ sends DONE to the supervisor.
\item $p$ sends REJECT($R$) to the supervisor for some subset of the predecessors of $v$.
\end{itemize}
Whenever the supervisor asks an adversarial worker $p$ to send its output to some worker $q$, $p$ might either send the correct output to $q$, some wrong output to $q$, or nothing to $q$ which we simply consider as a wrong output as well.

We want to determine a probability bound for the runtime to be at least $D+s$ for some sufficiently large $s$, no matter what kind of strategy the adversary is using. First, we show that a runtime of at least $D+s$ implies the existence of a so-called $s$-delay sequence, which corresponds to a walk through the task graph.

Consider any execution of the supervised computation that takes at least $D+s$ time. We will construct a corresponding $s$-delay sequence witnessing that. The delay sequence starts with node $u_{D+s}$ representing a task that is executed at time $D+s$. From this task, the delay sequence follows the wavefront of the computation backwards in time. The construction takes into account that for each node $u_i$ representing a task $v$ that is executed at time $i$, one of the three cases must have happened in round $i-1$ so that the supervisor added $v$ to the wavefront and therefore decided to execute $v$: (1) the execution of one of the predecessors of $v$ ended with a DONE message, (2) a successor of $v$ sent a REJECT message including $v$ to the supervisor, or (3) the worker executing $v$ did not send anything to the supervisor. Note that $v$ cannot be executed due to any non-successor sending a REJECT message in the previous round because if it is not at the root of a subtree pruned from $F$ due to a REJECT message, at least one of its predecessors must have been pruned from $F$ as well, which means that $v$ cannot belong to the wavefront. The delay sequence then moves from $u_i$ to $u_{i-1}$ representing
\begin{itemize}
\item any predecessor of $v$ sending a DONE in the previous round in case (1),
\item any successor of $v$ sending a REJECT including $v$ in the previous round in case (2), and
\item $v$ in case (3).
\end{itemize}
The delay sequence ends with a node $u_1$ representing an initial node that was executed by the source in the first round. Thus, the delay sequence recorded that way consists of $D+s-1$ edges (that either represent edges in the task graph or self-loops) and $D+s$ nodes. For each node $u_i$ of the delay sequence we record the set $(v_i,s_i,s'_i)$ where $v_i$ is the task associated with $u_i$, $s_i \in \{H,A\}$ is the state of the worker $p$ assigned to $v_i$ at time $i$, and $s'_i \in \{H,A,-\}$ is the state of the worker last chosen for $v_i$ before $p$ (which is "-" if $p$ was the first worker chosen for $v_i$).

Next, we investigate the state properties an $s$-delay sequence must have to witness a runtime of at least $D+s$. First, consider the case that the edge from $u_i$ to $u_{i-1}$ represents a self-loop, i.e., $v_{i-1}=v_i$. Then the worker chosen for task $v_{i-1}$ at time $i-1$ must have been adversarial, i.e., $s_{i-1}=A$. Next, consider the case that the edge from $u_i$ to $u_{i-1}$ is a \emph{rollback edge}, i.e., $v_i$ is a predecessor of $v_{i-1}$ in the task graph. Certainly, for a rollback to occur, either the worker that executed $u_{i-1}$ at time $i-1$ or the worker that last executed $u_i$ before round $i$ must have been adversarial, which implies that either $s_{i-1}=A$ or $s'_i=A$. 

We call an $s$-delay sequence \emph{valid} if (1) the tasks $v_1,\ldots,v_{D+s}$ associated with it satisfy the property that either $v_{i+1}=v_i$, $v_{i+1}$ is a predecessor of $v_i$, or $v_{i+1}$ is a successor of $v_i$ in the task graph for all $i$, and (2) its states $(s_i,s'_i)$ satisfy the properties above for all occurrences of self-loops and rollback edges for all $i$. It follows that if the runtime of the supervised computation is at least $D+s$ then there must exist a valid $s$-delay sequence, or in other words, if there does not exist a valid $s$-delay sequence then the runtime of the supervised computation is less than $D+s$. Therefore, it remains to bound the probability of a valid $s$-delay sequence to exist.

We first bound the number of different choices of $v_1,\ldots,v_{D+s}$ for an $s$-delay sequence.

An $s$-delay sequence consists of $D+s-1$ edges, each of which can be either one of at most $d$ edges to a predecessor, one of at most $d$ edges to a successor, or a self-loop in the task graph, resulting in at most $2d+1$ possible choices. Furthermore, the delay sequence may start at anyone of the tasks in $G$. Therefore, there are at most 
\[
    n \cdot (2d+1)^{D+s-1}
\]
ways of fixing the tasks $v_1,\ldots,v_{D+s}$.

To bound the probability of an $s$-delay sequence with fixed $v_1,\ldots,v_{D+s}$ to be valid, we associate a probability with each rollback edge and each self-loop of the $s$-delay sequence. Recall that every node $u_i$ in the delay sequence witnesses an execution of $v_i$ at time $i$. First, consider the case that the edge from $u_i$ to $u_{i-1}$ represents a self-loop. Then the worker chosen for task $v_{i-1}$ at time $i-1$ must have been adversarial, which only happens with probability $\beta$. Next, consider the case that the edge from $u_i$ to $u_{i-1}$ represents a rollback from $u_{i-1}$. Recall that for a rollback to occur, either the worker that executed $v_{i-1}$ at time $i-1$ or the worker that last executed $v_i$ before round $i$ must have been adversarial. The probability for the worker that executed $v_{i-1}$ at time $i-1$ to be adversarial is $\beta$. Thus, it remains to consider the case that the worker that last executed $v_i$ before round $i$ is adversarial. Let the corresponding task-time pair be called $(v_i,t)$. While the task-time pairs associated with the nodes of the delay sequence are clearly disjoint, and therefore we can get independent probabilities for these whenever an adversary is involved in such a task, this is not obvious for the task-time pair $(v_i,t)$. There are two cases to consider.

Case 1: $(v_i,t)$ is identical to the task-time pair associated with $v_j$ for some $j<i$. If that task-time pair was counted as adversarial due to not responding to the supervisor or being involved in a rollback then the supervisor would have chosen a new worker for that task, and this new worker (or some worker chosen later) would have been the worker in question for $(v_i,t)$, contradicting our assumption.

Case 2: $(v_i,t)$ is identical to some other $(v_j,t')$ that witnessed a rollback from $v_{j-1}$ for some $j\not=i$. If $j<i$ then $v_j$ would have been removed from the set of finished tasks at time $j$, which means that the supervisor would have asked a worker later assigned to $v_i$ to send its output to the task-time pair $(v_{i-1},i-1)$. The case $j>i$ is analogous. Thus, also in this case, $(v_i,t)$ must be disjoint with other task-time pairs witnessing an adversarial behavior in the delay sequence.

Since the $s$-delay sequence starts at a task $v_{D+s}$ at level at most $D$ and ends at task $v_1$ at level at least 0, it must consist of at least $\nicefrac{s}{2}$ edges that are either rollback edges or self-loops. Since there are two possible choices for the task-time pair causing a rollback, this results in the following probability bound
\[
    \Pr[\text{a specific $s$-delay sequence is valid}]\leq (2\beta)^{\nicefrac{s}{2}}
\]

Finally, we bound the probability that the supervised computation takes at least $D+s$ rounds for $s\geq 2(D+\nicefrac{c \log n}{\log (2d+1)})/\epsilon$, where $c>1$ and $\epsilon>0$ are constants.
\begin{align*}
    \Pr[\text{runtime $\ge D+s$}]&\leq\Pr[\text{There exists an $s$-delay sequence}]\\
    &\leq n \cdot (2d+1)^{D+s-1} \cdot (2\beta)^{\nicefrac{s}{2}}\\
    &\leq n \cdot (2d+1)^{D+s-1} \cdot (2d+1)^{-(1+\epsilon/2)s} \tag{$\beta\leq(\frac{1}{2(2d+1)})^{2+\epsilon}$}\\
    &\leq n \cdot (2d+1)^{D-\epsilon s/2} \\
    &\leq n \cdot (2d+1)^{-\frac{c\log n}{\log (2d+1)}} \tag{$s\geq 2(D+\nicefrac{c\log n}{\log (2d+1)})/\epsilon$}\\
    &= n^{-c+1}
\end{align*}
Thus, the runtime is at most $\bigO((D+\log_d n)/\epsilon)$, w.h.p.
\end{proof}

Thus, the runtime is asymptotically optimal if $D \ge \log n$. There are DAGs where $\beta<1/d$ is required to finish within any polynomial in $D$ many rounds so that our bound on $\beta$ is not far from optimal. To see that, consider a leveled network of depth $D$ in which every level contains $d$ tasks and the tasks of two consecutive levels form a complete bipartite network. If the adversary uses the strategy to always reject all inputs, a single adversarial peer can roll back the computation of an entire level, and the probability of selecting no adversarial peer in a level is $(1-\beta)^d \le e^{-\beta d}$. Thus, standard arguments from biased random walks imply that if $\beta \ge 1/d$, the number of rounds needed to finish the computation is exponential in $D$ on expectation.

The I/O work of the source and the target appears to be hard to bound, but we can use a trick to reduce it to an asymptotically optimal work by extending $G$ in an appropriate way. Let $c\ge 1$ be a sufficiently large constant. Consider the DAG $G'=(V',E')$ that results from $G$ by

\begin{itemize}
\item adding a directed \emph{initial list} of $k=c\log n$ nodes $u_1(v),\ldots,u_k(v)$ to each initial node $v \in V$ in a sense that $u_1(v)$ is the new initial node and $u_k(v)$ has an edge to $v$ and
\item adding a directed \emph{final list} of $\ell=D+\log n$ nodes $u_1(v),\ldots,u_{\ell}(v)$ to each final node $v \in V$ in a sense that $v$ has an edge to $u_1(v)$ and $u_{\ell}(v)$ is the new final node.
\end{itemize}

All that the initial list is doing is forwarding the input of the source from one node to the next one, and 
all that the final list is doing is forwarding the output of the original final node from one node to the next one. 

Since the construction only increases the depth of $G$ by at most $D+(c+1)\log n$, the asymptotic runtime bound of Lemma~\ref{th:dag-runtime} remains the same. It remains to bound the I/O work of the source and the target. We first make use of Lemma~\ref{lem:advsupervisor} to obtain the following result:

\begin{lemma} \label{th:dag_sendingwork}
If $\beta \le 1/12$ then under any adversarial strategy, the source just needs to send the input $O(1)$ times to any initial node in $G'$ on expectation.
\end{lemma}
\begin{proof}
Due to Lemma~\ref{th:dag-runtime}, it takes at most $O(n)$ rounds, w.h.p., to finish the computation in $G'$. Consider some initial list $L$. Clearly, the worst case for the expected number of times the source needs to send the input to the initial node of $L$ is reached if the output of a worker assigned to the last task of $L$ is always rejected. This is exactly the situation of Lemma~\ref{lem:advsupervisor}, which implies that even in such a situation the source will only send the input to the initial node of $L$ $O(1)$ times on expectation within $O(n)$ rounds.
\end{proof}

Next, we bound the number of times the target receives the output from a final list.

\begin{lemma}\label{th:dag_receivingwork}
If $\beta\leq 1/3$ then for any adversarial strategy, the target just needs to receive the output $O(1)$ times from any final node in $G'$ on expectation.
\end{lemma}
\begin{proof}
Consider any final list $L$ and let $v_1$ be its first node. When considering $L$ in an isolated fashion and going through the proof of Lemma~\ref{th:line-output}, it turns out for $\beta \le 1/3$ that no matter whether the source delivers a correct or wrong input, the expected number of times the target receives the output is $O(1)$, because both extremes $a_0$ and $h_0$ are equal to $O(1)$. In other words, no matter which input is given to $v_1$ in $G'$, the number of times the target receives an output from the final node of $L$ is $O(1)$. However, it might happen that rollbacks in some part of the DAG might cause a complete rollback of the computation in $L$. If two of such rollbacks are less than $D+\log n$ rounds apart, then the computation within $L$ could not have reached the target. Since the computation in $G'$ will finish in $O(D+\log n)$ rounds, w.h.p., there can be at most $O(1)$ occasions where the computation in $L$ reaches the target, w.h.p., and for each of these occasions the target will only receive the output $O(1)$ times on expectation, as argued above, which implies the lemma.
\end{proof}
 
As the target only accepts a correct output, the target must have received correct outputs from all final nodes once the computation terminates.

Combining this fact with the lemmas above proves \Cref{thm:arb}.

\section{Supervised Matrix Multiplication} \label{sec:matrix}

Consider the well-known matrix multiplication problem: Given two $m \times m$-matrices $A$ and $B$, compute $C=A \cdot B$. Standard approaches for multiplying the matrices in parallel with $n$ processes are to decompose $A$ and $B$ into stripes or squares and to distribute these in an appropriate way among the given processes in order to compute $C$ (e.g., \cite{Cannon69}).

We employ the stripe approach since it most elegantly demonstrates our supervised computing concept. In the following, $n$ denotes the number of matrix multiplication tasks instead of the total number of tasks in a task graph since it is easier to describe the task graph this way. However, since its size is equal to $\Theta(n \log n)$, our probability bounds still hold w.r.t. the size of the task graph. Let $m \ge \sqrt{n} \log n$ and suppose, for simplicity, that $k=\sqrt{n}$ is a power of 2 that divides $m$. For every $i \in \{1,\ldots,k\}$ let $A_i$ be the $m \times (m/k)$-matrix consisting of the rows $1+(i-1)m/k,\ldots,i \cdot m/k$ of $A$ and $B_i$ be the $(m/k) \times m$-matrix consisting of the columns $1+(i-1)m/k,\ldots,i \cdot m/k$ of $B$.

Consider any two matrices $A_i$ and $B_j$. There is a well-known randomized algorithm for the test if $A_i \cdot B_j = C_{i,j}$ for some given $C_{i,j}$ called \emph{Freivald's algorithm} \cite{Freivalds77}:

\begin{algorithm}
    \caption{Freivald's Algorithm($A_i$, $B_j$, $C_{i,j}$)}\label{alg:freivald}
    
    $r \gets $ random vector $\in\{0,1\}^{\nicefrac{m}{k}}$\;
    
    $x \gets B_j \cdot r$\;
    $y \gets A_i \cdot x$\;
    $z \gets C_{i,j} \cdot r$\;
    \If{$y=z$}
        {Return true\;}
    \Else
        {Return false\;}
    
\end{algorithm}

The following lemma directly follows from the proof in \cite{Freivalds77}:

\begin{lemma}\label{lem:freivalds}
Let $A_i\in\mathcal{M}_{m/k,m}(\mathbb{R})$, $B_j \in\mathcal{M}_{m,m/k}(\mathbb{R})$ and $C_{i,j}\in\mathcal{M}_{m/k,m/k}(\mathbb{R})$ be three matrices with $A_i \cdot B_j \not= C_{i,j}$ and let $r\in \mathbb{R}^{m/k}$ be a vector chosen uniformly and independently at random. Then $\Pr[A_i \cdot B_j \cdot r = C_{i,j} \cdot r] \le 1/2$.
\end{lemma}

Recall that $k=\sqrt{n}$ is an integer, which means that we want to compute $n$ many matrices $C_{i,j}$. We do that with the DAG $G=(V,E)$ as shown in Figure~\ref{fig:matrixmult}. The parts of this DAG are defined as follows:
\begin{itemize}
\item $I_{X,j}$, $X \in \{A,B\}$ and $j \in \{1,\ldots,k\}$: These are lists of length $c\log n$ for a sufficiently large constant $c$.
\item $T_{X,i}$, $X \in \{A,B\}$ and $i \in \{1,\ldots,k\}$: These are complete binary trees of depth $\log k$, i.e., each tree has $k$ leaves.
\item $O_{i,j}$, $i,j \in \{1,\ldots,k\}$: These are lists of length $c \log n$.
\end{itemize}
Each $v_{i,j}$, $i,j \in \{1,\ldots,k\}$, has two incoming edges: from the $j$-th leaf of $T_{A,i}$ and the $i$-th leaf of $T_{B,j}$. Thus, every node of $G$ has indegree and outdegree at most 2.

\begin{figure}[ht!]
    \begin{center}
        \includegraphics[width=0.6\textwidth]{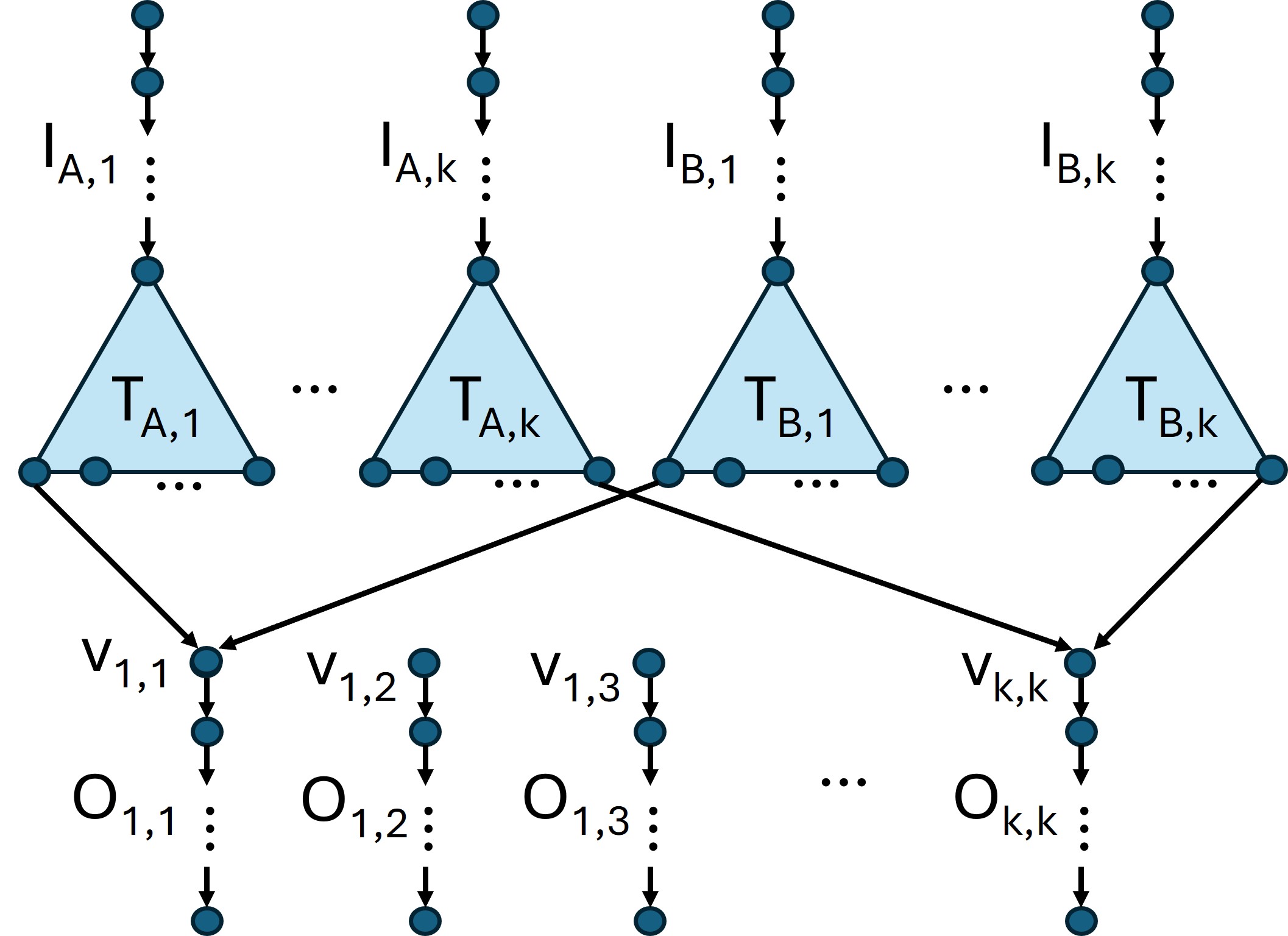}
    \end{center}
    \caption{The DAG for  matrix multiplication.}
    \label{fig:matrixmult}
\end{figure}

The supervisor $s$ uses the following strategy to compute $A \cdot B$ with the help of $G$. For every $i \in \{1,\ldots,k\}$ and $X \in \{A,B\}$, $s$ will ask the starting node of $I_{X,i}$ to fetch $X_i$ from the source. The role of the nodes in each $I_{X,i}$ is to simply forward $X_i$ to the next node in the list, and the role of each inner node of $T_{X,i}$ is to forward $X_i$ to its two children. Each node $v_{i,j}$ is supposed to receive $A_i$ and $B_j$, compute $C_{i,j}=A_i \cdot B_j$, and send $(A_i,B_j,C_{i,j})$ to the starting node of $O_{i,j}$. Each node in $O_{i,j}$ is supposed to simply forward $(A_i,B_j,C_{i,j})$ to the next node.

If the source precomputes the hashes $h(A_i)$ and $h(B_i)$ for some collision-resistant hash function $h$ and these hashes are fetched by the supervisor at the beginning, then the workers assigned to nodes in $I_{X_i}$, $T_{X,i}$, $v_{i,j}$, and $O_{i,j}$ can easily verify the correctness of $A_i$ and $B_j$ with the help of the corresponding hash provided by the supervisor. On top of that, an honest worker in $O_{i,j}$ can run Freivalds' test $\tau$ times on $(A_i,B_j,C_{i,j})$ in order to reduce the probability of accepting $C_{i,j}$ even though $A_i \cdot B_j \not=C_{i,j}$ to at most $1/2^\tau$. If an honest worker accepts $C_{i,j}$, it sends DONE and $h(C_{i,j})$ to the supervisor. Whenever a worker associated with the final node of some $O_{i,j}$ sends $C_{i,j}$ to the target, the target requests $h_{i,j}$, the majority value of the one-way hashes sent to the supervisor by the workers currently assigned to nodes in $O_{i,j}$, and tests if $h(C_{i,j})=h_{i,j}$. If not, it rejects the output and otherwise it accepts it. 

Note that Lemma~\ref{lem:line-majority} implies that once an output is sent to the target for the first time, the majority of the workers assigned to nodes in $O_{i,j}$ is honest, w.h.p., which means that w.h.p. the supervisor will send a hash $h_{i,j}$ to the target that indeed satisfies $h(C_{i,j})=h_{i,j}$. Thus, when the computation terminates, the target can correctly assemble $C$, w.h.p.

We start by stating the sections main theorem.
\begin{theorem}
\label{thm:matrix}
If $\beta + 1/2^\tau \le 1/200$ and $m \ge \sqrt{n} \log n$ then for any adversarial strategy, the supervised matrix multiplication terminates correctly in $O(\log n)$ rounds, w.h.p. Furthermore, the total computational work of the workers is at most $O(m^3)$ on expectation, the maximum computational work for an individual task is $O(m^3/n)$, the total work by the source and the target is at most $O(m^2)$ on expectation and the total work of the supervisor is at most $O(n\log^3n)$.
\end{theorem}

We continue by bounding the algorithm's runtime in communication rounds. If we consider a worker to be adversarial if an adversarial worker is chosen or an honest worker incorrectly accepts some matrix product, Lemma~\ref{th:dag-runtime} immediately implies (choose $\epsilon>0$ so that $(\nicefrac{1}{2(2d+1)})^{2+\epsilon}=1/200$):

\begin{lemma} \label{th:mm-runtime}
If $\beta+1/2^\tau \le 1/200$ then for any adversarial strategy, the supervised matrix multiplication terminates within $O(\log n)$ rounds, w.h.p.
\end{lemma}

Next, we bound the work of the source, the target, the supervisor, and the workers. From Lemma~\ref{th:dag_sendingwork} the following result follows.

\begin{corollary}
If $\beta \le 1/12$ then for any adversarial strategy and for any of the lists $I_{X,i}$, the source just needs to send the input $O(1)$ times on expectation.
\end{corollary}

Thus, the expected work of the source to deliver the inputs is $O(m^2)$.

Applying Lemma~\ref{th:dag_receivingwork} to the lists $O_{i,j}$ we also get:

\begin{corollary}
If $\beta \le 1/3$ then for any adversarial strategy, the expected work for the target to receive the outputs from the workers is $O(m^2)$.
\end{corollary}

\begin{lemma}
    If $\beta \le 1/12$ then for any adversarial strategy, the expected work for the supervisor is $O(n\log^3n)$.
\end{lemma}
\begin{proof}
    There are $O(n\log n)$ tasks in total in the task graph. Due to the bound on the number of communication rounds from \Cref{th:mm-runtime}, the supervisor assigns no more than $O(\log n)$ workers to each task. As it makes only a constant number of introductions for each of them, this yields a total work of $O(n\log^2n)$ for assigning and introducing workers. On top of that, the supervisor has to provide a constant number of hash values for each task assignment. Since the hash size should be $O(\log^2 n)$ or larger to be collision-resistant, the work for sending the hash values is $O(n\log^3 n)$.    Thus, in total, the work of the supervisor is bounded by $O(n\log^3n)$.
\end{proof}

Next, we show the following work bound for the honest workers, which is asymptotically optimal.

\begin{lemma}
If $\beta + 1/2^\tau \le 1/200$ and $m \ge \sqrt{n} \log n$ then the total expected computational work for the workers is $O(m^3)$ and the maximum computational work for an individual task is $O(m^3/n)$.
\end{lemma}
\begin{proof}
Suppose for now that all workers are honest. Then we get:
\begin{itemize}
\item The total work for each $I_{X,i}$ and $O_{i,j}$ is $O((m/k) \cdot m \log n) = O((m^2/\sqrt{n}) \log n)$.
\item The total work for each $T_{X,i}$ is $O(k (m/k) \cdot m) = O(m^2)$.
\item The work for each $v_{i,j}$ is $O(m \cdot (m/k)^2) = O(m^3/n)$.
\end{itemize}
Summing up these work bounds, we obtain a total work of
\[
  O(k \cdot [(m^2/\sqrt{n}) \log n] + k \cdot m^2 + k^2 \cdot m^3/n) = O(m^3)
\]
if $m \ge \sqrt{n}$. Now, consider the case that we have a $\beta$-fraction of adversarial workers. If $\beta \le 1/12$ then Lemma~\ref{th:line-middle} implies that the expected amount of work that honest workers have to invest for all $v_{i,j}$ is still $O(m^3)$. From the runtime bound in Lemma~\ref{th:mm-runtime} we know that each of the remaining nodes can be executed at most $O(\log n)$ times if $\beta + 1/2^\tau \le 1/200$, which gives a worst-case work bound of
\[
  O(k \cdot [(m^2/\sqrt{n}) \log^2 n] + k \cdot m^2 \log n) = O(m^2 \sqrt{n} \log n)
\]
for the remaining nodes. This is still bounded by $O(m^3)$ if $m \ge \sqrt{n} \log n$. The maximum computational work for an individual task follows from the fact that $O(m^3/n)$ work is needed for each $v_{i,j}$ and the work for the other tasks is much lower.
\end{proof}

From the above lemmas and corollaries, we can conclude \Cref{thm:matrix}

\section{Supervised Mergesort} \label{sec:mergesort}

In this section, we consider the problem that the source has $m$ data items that it wants to sort with the help of the workers where at most a $\beta$-fraction of the workers is malicious.

Various sorting networks have already been proposed for parallel sorting like Batcher's odd-even and bitonic sorting networks \cite{Batcher68} or the AKS sorting network \cite{AjtaiKS83}. However, the proposed sorting networks with a simple topology have a depth of $\Omega(\log^2 n)$ while the AKS sorting network is considered to be impractical because the constant factor in its depth bound of $\bigO(\log n)$ is extremely large \cite{Paterson90}. Furthermore, these sorting networks are just based on simple comparators and not meant for adversarial behavior. Thus, we employ a different network with more complex operations within the nodes that resembles the classical mergesort algorithm, which might be of independent interest.

The goal of this section is to prove the following theorem:
\begin{theorem} \label{th:merge}
If $\beta\leq 1/200$ then for any adversarial strategy, the supervised mergesort algorithm terminates within $\bigO(\log n)$ rounds, w.h.p. Furthermore, the source and the target just need $O(m)$ computational and communication work on expectation, the supervisor needs $O(n\log n)$ work, w.h.p, the total computational work performed by the workers is $O(m \log m)$ and the total communication work is $O(m \log n)$ on expectation, and the maximum work of a task is $O((m/n) (\log m/\log n + \log n))$, w.h.p. 
\end{theorem}

To achieve this, we first describe and analyze the supervised mergesort algorithm without malicious activities. Afterwards, we consider the malicious case and analyze how it can be handled by the supervisor, the source, the target, and the honest workers.

\subsection{Honest Setting}
    
The task graph for this problem is a leveled network where each layer has the same number of tasks. As this number is more relevant for our bounds than the size of the task graph, we use $n$ to denote that number in this section. Let $m=\Omega(n\log n)$. For simplicity, we assume that $n$ divides $m$, and $n$ is a power of $2$.

We start with a high-level description of the algorithm. The data items are partitioned into $n$ sets of size $\nicefrac{m}{n}$, and each of the $n$ initial nodes of the task graph is supposed to sort one of these sets. 

The remainder of the task graph is similar to the merging procedure performed by mergesort, i.e., we will repeatedly merge the sets of data items in each layer, resulting in a doubling of their size. As we want to keep the merging work of the tasks roughly equal, we also double the number of tasks responsible for each set in each layer. After $\log n$ rounds of merging, we have a single sorted set of data items distributed among $n$ tasks and can send the output to the target.

\paragraph*{Detailed Description}

The input $I$ is preprocessed as follows. First, the source randomly permutes the data items in $I$ and assigns an index to each one of them based on its position in that permutation. The index can be used as a tie breaker to ensure that all data items are pairwise distinct, but they will also turn out to be useful for the verification. Next, the source picks a subset of $n$ data items uniformly at random out of $I$, which we call \emph{quantiles}, and sorts them so that $q_0 < q_1 < \ldots < q_{n-1}$. In addition to sending the quantiles to the superviser, for every $j \in \{0,\ldots,n-1\}$, the source sends the data items with indices $j\cdot\nicefrac{m}{n}+1,\dots,(j+1)\cdot\nicefrac{m}{n}$ and quantile $q_{r(j)}$ to initial node $j$, where $r:\{0,1\}^{\log n} \to \{0,1\}^{\log n}$ is the bit reversal permutation, i.e., it maps the binary representation $(b_1\dots b_{\log n})$ of a quantile's index to the binary representation $(b_{\log n}\dots b_1)$. An example of this permutation for $n=8$ is presented in \Cref{fig:permutation}. 

    \begin{figure}[ht!]
        \begin{center}
            \begin{tikzpicture}
    \foreach \y [count=\yi] in {{000,001,010,011,100,101,110,111},{000,100,010,110,001,101,011,111}} {
        \foreach \x [count=\xi] in \y {
            \node[text=black!50] (\x_\yi) at (\xi,1-1.5*\yi) {\x};
        }
    }
    \foreach \y [count=\yi] in {{0,1,2,3,4,5,6,7},{0,4,2,6,1,5,3,7}} {
        \foreach \x [count=\xi] in \y {
            \node (\x_\yi) at (\xi,2.5-2.5*\yi) {\textbf{\x}};
        }
    }
    \foreach \x in {000,001,010,011,100,101,110,111} {
        \draw[->] (\x_1.south) -- (\x_2.north);
    }
    \node[anchor=west] at (8.5,-0.25) {Quantiles};
    \node[anchor=west] at (8.5,-2.25) {Tasks};
\end{tikzpicture}
        \end{center}
        \caption{An example of the bit reversal permutation with $8$ quantiles. Initial node $i$ receives quantile $r(i)$.}
        \label{fig:permutation}
    \end{figure}
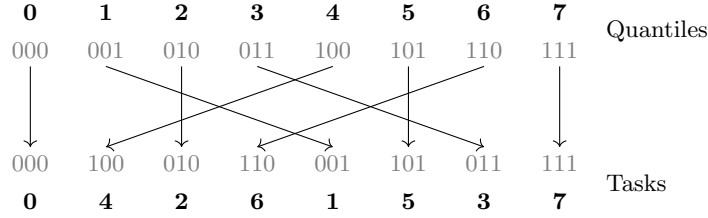

The task graph consists of two parts --- a sorting part and a merging part. The sorting part consists of the first layer of the task graph, called layer 0. The task of each node in that layer is to sort the data items assigned to that node.

The remaining nodes represent the merging part, where we perform a distributed merging of the sets of data items sorted in the sorting part. The merging part consists of $\log n$ layers of $n$ nodes each, numbered from 1 to $\log n$. For every $i \in \{0,\ldots,\log n\}$, let the nodes in layer $i$ be called $v_{i,0},\ldots,v_{i,n-1}$. In each layer $i$, the nodes are partitioned into $n/2^i$ \emph{segments} $S_{i,0},\ldots,S_{i,n/2^i-1}$, where segment $S_{i,j}$ consists of the nodes $v_{i,j \cdot 2^i},\ldots,v_{i,(j+1) \cdot 2^i-1}$. For each $i$ and $j$, our goal is to merge the sorted sets stored in $S_{i,2\cdot j}$ and $S_{i,2\cdot j+1}$ to a sorted set stored in $S_{i+1,j}$. From layer 0 to 1, this is done in the following way.

Note that $q_{r(2j)}<q_{r(2j+1)}$ for every $j$ since $r(2j+1)=r(2j)+n/2$ and we assume the quantiles to be sorted. For each $j$, the supervisor additionally informs $v_{0,2j}$ about $q_{r(2j+1)}$ and $v_{0,2j+1}$ about $q_{r(2j)}$. This allows $v_{0,2j}$ and $v_{0,2j+1}$ to split their sorted sets of data items into two sets belonging to the ranges $[q_{r(2j)},q_{r(2j+1)})$ and $[q_{r(2j+1)},+\infty) \cup (-\infty,q_{r(2j)})$ and to send these to $v_{1,2j}$ and $v_{1,2j+1}$ in segment $S_{i+1,j}$ so that $v_{1,2j}$ contains all data items of $v_{0,2j}$ and $v_{0,2j+1}$ in the range $[q_{r(2j)},q_{r(2j+1)})$ while $v_{0,2j+1}$ contains all data items of $v_{0,2j}$ and $v_{0,2j+1}$ in the range $[q_{r(2j+1)},+\infty) \cup (-\infty,q_{r(2j)})$. Certainly, $v_{1,2j}$ and $v_{1,2j+1}$ just need to merge two sorted sets to make sure that they store a sorted set of data items in the range assigned to them.

For every $i \ge 1$, the merging from layer $i$ to $i+1$ works as follows. We start with an assumption on the setup in layer $i$, which certainly holds for layer 1 from the description above and can then be shown to inductively hold for later layers:

Suppose that the quantiles in $S_{i,j}$ (containing $v_{i,j \cdot 2^i},\ldots,v_{i,(j+1) \cdot 2^i-1}$) are $q_{r(j \cdot 2^i)},\ldots,q_{r((j+1) \cdot 2^i-1)}$, which are equal to the set of quantiles $\{ q_{r(j\cdot 2^i)+k\cdot n/2^i} \mid k \in \{0,\ldots,2^i-1\} \}$ and the same quantiles as those assigned to the initial nodes $v_{0,j \cdot 2^i},\ldots,v_{0,(j+1) \cdot 2^i-1}$. To simplify notation, let $q(i,j,k) = q_{r(j\cdot 2^i)+k\cdot n/2^i}$. Since $q_0<q_1<\ldots<q_{n-1}$, it certainly holds that $q(i,j,k)<q(i,j,k+1)$ for every $k$. Suppose that for every $k\in \{0,\ldots,2^i-1\}$, the node $v_{i,j \cdot 2^i+k}$ in $S_{i,j}$ stores quantiles $q(i,j,k)$ and $q(i,j,k+1 \mod 2^i)$. Further, suppose that each node $v_{i,j \cdot 2^i+k}$ with $k<2^i-1$ stores all data items in the range $[q(i,j,k),q(i,j,k+1))$ that have been assigned to the initial nodes $v_{0,j \cdot 2^i},\ldots,v_{0,(j+1) \cdot 2^i-1}$ and node $v_{i,(j+1) \cdot 2^i-1}$ stores all such data items in the range $[q(i,j,2^i-1),+\infty) \cup (-\infty,q(i,j,0))$. This is also the goal that we want to reach for layer $i+1$. 

To achieve that goal, we make use of the fact that the quantiles of $S_{i,2j}$ and $S_{i,2j+1}$ are perfectly interleaved, i.e., $q(i,2j,k) < q(i,2j+1,k)$ and $q(i,2j+1,k) < q(i,2j,k+1)$ for all $k$ (see also Figure~\ref{fig:task_graph_honest}). This follows from the fact that $r(2j\cdot 2^i)+k\cdot n/2^i < r((2j+1)\cdot 2^i)+k\cdot n/2^i$, $r((2j+1)\cdot 2^i)+k\cdot n/2^i < r(2j\cdot 2^i)+(k+1)\cdot n/2^i$ and the quantiles are sorted. For each $k$, the $k$th node in $S_{i,2j}$ splits its set of data items into two according to quantile $q(i,2j+1,k)$ (additionally provided by the supervisor) and sends the part $<q(i,2j+1,k)$ to the $2k$th node in $S_{i+1,j}$ while sending the part $\ge q(i,2j+1,k)$ to the $2k+1$th node in $S_{i+1,j}$. Similarly, for each $k<2^i-1$, the $k$th node in $S_{i,2j+1}$ splits its set of data items into two according to quantile $q(i,2j,k+1)$ (additionally provided by the supervisor) and sends the part $<q(i,2j,k+1)$ to the $2k+1$th node in $S_{i+1,j}$ while sending the part $\ge q(i,2j,k+1)$ to the $2k+2$th node in $S_{i+1,j}$. Special treatment needs to be given to node $k=2^i-1$ (the last node) of $S_{i,2j+1}$ since it needs to send its set of data items in the range $[q(i,2j+1,k),+\infty) \cup (-\infty,q(i,2j,0))$ to the last node of $S_{i+1,j}$ and the set of data items in the range $[q(i,2j,0),q(i,2j+1,0))$ to the first node of $S_{i+1,j}$. In that way, each node $v_{i+1,j \cdot 2^{i+1}+k}$ with $k \in \{0,\ldots,2^{i+1}-2\}$ will store all data items in the range $[q(i+1,j,k),q(i+1,j,k+1))$ that have been assigned to the initial nodes $v_{0,j \cdot 2^{i+1}},\ldots,v_{0,(j+1) \cdot 2^{i+1}-1}$ and node $v_{i+1,(j+1) \cdot 2^i-1}$ will store all such data items in the range $[q(i+1,j,2^{i+1}-1),+\infty) \cup (-\infty,q(i+1,j,0))$, which will satisfy the assumptions made above for layer $i$ also for layer $i+1$.

An example of the resulting task graph is depicted in \cref{fig:task_graph_honest}. As can be seen there and the description above, the indegree and outdegree of every node of the task graph is at most 2.

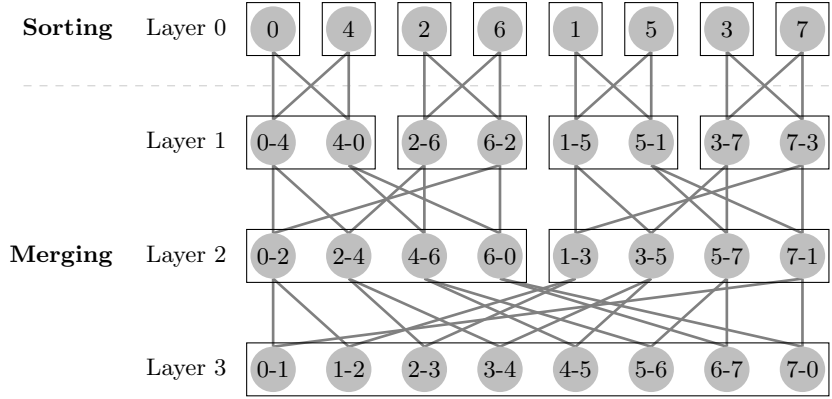
\begin{figure}[ht!]
    \begin{center}
        \begin{tikzpicture}

    \draw[dashed,black!20] (-2.3,-.75) -- (8.5,-.75);

    \foreach \connections in {{1/1/2},{2/1/2},{3/3/4},{4/3/4},{5/5/6},{6/5/6},{7/7/8},{8/7/8}} {
        \foreach \a/\b/\c in \connections {
            \draw[black!50,line width=1px] (\a,-0.3) -- (\b,-1.2);
            \draw[black!50,line width=1px] (\a,-0.3) -- (\c,-1.2);
        }
    }
    \foreach \connections in {{1/1/2},{2/3/4},{3/2/3},{4/1/4},{5/5/6},{6/7/8},{7/6/7},{8/5/8}} {
        \foreach \a/\b/\c in \connections {
            \draw[black!50,line width=1px] (\a,-1.8) -- (\b,-2.7);
            \draw[black!50,line width=1px] (\a,-1.8) -- (\c,-2.7);
        }
    }
    \foreach \connections in {{1/1/2},{2/3/4},{3/5/6},{4/7/8},{5/2/3},{6/4/5},{7/6/7},{8/1/8}} {
        \foreach \a/\b/\c in \connections {
            \draw[black!50,line width=1px] (\a,-3.3) -- (\b,-4.2);
            \draw[black!50,line width=1px] (\a,-3.3) -- (\c,-4.2);
        }
    }
    
    \foreach \y in {1,2,3,4} {
        \foreach \x [count=\xi] in {1,5,3,7,2,6,4,8} {
            %\node (\x_\y) at (\xi,1-\y) {\x};
            \node[circle,fill=black!25,minimum size=.6cm] (\x_\y) at (\xi,1.5-1.5*\y) {};
        }
    }
    
    \foreach \y [count=\yi] in {{0,4,2,6,1,5,3,7},{0-4,4-0,2-6,6-2,1-5,5-1,3-7,7-3},{0-2,2-4,4-6,6-0,1-3,3-5,5-7,7-1},{0-1,1-2,2-3,3-4,4-5,5-6,6-7,7-0}} {
        \foreach \x [count=\xi] in \y {
            \node at (\xi,1.5-1.5*\yi) {\x};
        }
    }

    \foreach \y [count=\yi from 0] in {{1/1,2/2,3/3,4/4,5/5,6/6,7/7,8/8},{1/2,3/4,5/6,7/8},{1/4,5/8},{1/8}} {
        \foreach \x/\xi in \y {
            \draw ($(\x,-1.5*\yi)+(-0.35,-0.35)$) rectangle ($(\xi,-1.5*\yi)+(0.35,0.35)$);
        }
    }

    \foreach \y in {0,1,2,3} {
        \node[anchor=east] at (0.5,-1.5*\y) {Layer \y};
    }

    \node[anchor=east] at (-1,-3) {\textbf{Merging}};
    \node[anchor=east] at (-1,0) {\textbf{Sorting}};

\end{tikzpicture}
    \end{center}
    \caption{The task graph ($n=8$) of the supervised mergesort algorithm in the honest setting. In the first layer, the sets of data items sent by the source are sorted, 
    and in the subsequent layers, the sorted sets are merged.
    The grey circles correspond to tasks and the numbers in them denote the indices of the quantiles known to them and, based on these, the ranges they are responsible for. 
    The boxes represent the segments. 
    }
    \label{fig:task_graph_honest}
\end{figure}

The scheduling strategy used by the supervisor is the same as in the DAG case, i.e., it maintains a set $F$ of \emph{finished} tasks and remembers the last worker assigned to each task in $F$. In each round, for each executable task in the wavefront $W_F$, the supervisor picks a worker uniformly and independently at random to perform that task.

As its postprocessing strategy, the target concatenates the data items it received from each of the final tasks to obtain the output.

\paragraph*{Analysis}

The correctness of our supervised mergesort algorithm follows from the construction: After the algorithm's execution, we are left with a single segment in which the range assigned to node $v_{\log n, j}$ is $[q_j,q_{j+1 \mod n})$ for all $j$ and every node stores all data items of the input set that belong to its range.

Next, we consider the total work the workers have to perform.

\begin{lemma}\label{lem:msh:peerwork}
In the honest setting, the workers perform $O(m\log m)$ computational work and $O(m \log n)$ communication work during the execution of the supervised mergesort algorithm.
\end{lemma}
\begin{proof}
Each task $v_{0,j}$ needs $O((m/n) \log (m/n))$ time to sort its data items, which results in a  computation work of $O(m \log (m/n))$ in layer 0. In every layer $i \ge 1$, each task $v_{i,j}$ just needs to merge two sets of data items, which takes linear time. Since a total of $m$ data items are handled in each layer, this results in a computational work of $O(m)$ in each layer $i \ge 1$. Thus, the overall computational work is $O(m \log (m/n) + m \log n) = O(m \log m)$. In each layer, $m$ data items and $O(n)$ quantiles have to be received and sent, resulting in a total communication work of $O(m \log n)$.
\end{proof}

In addition to the total work, we also need to bound the maximum work per task. Each task of layer 0 just needs a computational work of $O((m/n) \log (m/n))$. However, the random selection of the quantiles can skew the work load among the merging tasks. In the following, we show that no merging task ever has to handle more than $\bigO((m/n)\log n)$ data items, w.h.p. Let $x_1<x_2< \ldots <x_m$ be the sorted sequence of data items in $I$. We start by proving that there cannot be a large subsequence that does not contain a quantile, w.h.p.

\begin{lemma}\label{lem:msh:workperpeer}
Let $S=\{x_i,\ldots,x_{i+s-1}\}$ be a fixed subsequence of $I$ of size $s$. If $s \ge c (m/n) \ln n$ for some $c \ge 1$, the probability that $S$ does not contain a quantile is at most $n^{-c}$.
\end{lemma}
\begin{proof}
For $s \ge c (m/n) \ln n$ we get
\begin{align*}
    \Pr[\text{No quantiles in }S] &= \frac{\text{\# quantile placements in }I\backslash S}{\text{\# quantile placements in }I} = \frac{{m-s\choose n}}{{m \choose n}}\\
    &= \frac{\nicefrac{\left(\prod_{i=0}^{n-1} m-s-i \right)}{n!}}{\nicefrac{\left(\prod_{i=0}^{n-1} m-i \right)}{n!}} = \prod_{i=1}^{n}\frac{m-s+1-i}{m+1-i}\\
    &\leq \prod_{i=1}^{n}\frac{m-s}{m} = \left(\frac{m-s}{m}\right)^n\\
    &\leq \left(1-\frac{s}{m}\right)^n \leq e^{-\frac{s}{m}n} \leq e^{-c \ln n} = n^{-c}
\end{align*}
\end{proof}

Thus, w.h.p., two consecutive quantiles are not more than $O((m/n)\log n)$ data items apart. It remains to show that this extends to the quantiles assigned to any given task when only considering the data items that are part of the set that task is responsible for.

\begin{lemma}  \label{lem:mergework}
In the honest setting, no task has to process more than $\bigO((m/n)\ln n)$ data items during the execution of the supervised mergesort algorithm, w.h.p.
\end{lemma}
\begin{proof}
The lemma trivially holds for layer 0. Recall that for each segment $S_{i,j}$ of layer $i\ge 1$, the quantiles $q(i,j,k)$ used in it are equal to the sorted sequence of quantiles $q_{r(j \cdot 2^i)},\ldots,q_{r((j+1) \cdot 2^i-1)}$. Thus, for each $k$, $q(i,j,k)$ and $q(i,j,k+1)$ represent two quantiles $q_{\ell}$ and $q_{\ell+n/2^i}$ for some $\ell$. We know from the previous lemma that for each $k \in \{0,\ldots,n/2^i-1\}$ there can be at most $c (m/n) \ln n$ many data items in $I$ between $q_{\ell+k}$ and $q_{\ell+k+1}$, w.h.p., if the constant $c>1$ is sufficiently large. Hence, the total number of data items in $I$ that are between $q_{\ell}$ and $q_{\ell+n/2^i}$ is at most $c(n/2^i)(m/n) \ln n = c (m/2^i) \ln n$, w.h.p., by a union bound. Next, we show that out of these data items only $O((m/n) \log n)$ many have to be handled by any task in segment $S_{i,j}$, w.h.p., which would prove the lemma. 

Let $M$ be the subset of data items in $I$ that are in the range $[q(i,j,k),q(i,j,k+1))$ for some $k$, and suppose that $|M|=c(m/2^i) \ln n$ for some constant $c$. Consider some fixed subset $N \subseteq M$ of size $c' (m/n) \ln n$ for some constant $c'>c$. Since the data items are randomly permuted and only a subset of $(m/n)2^i$ of these are part of $S_{i,j}$, the probability that all data items in $N$ are part of $S_{i,j}$ is equal to 
\[
    \frac{{(m/n) \cdot 2^i \choose |N|}}{{m \choose |N|}} \le \left( \frac{(m/n) \cdot 2^i}{m} \right)^{|N|} = \left( \frac{2^i}{n} \right)^{|N|}
\]    
On the other hand, there are at most 
\[
    {c(m/2^i) \ln n \choose |N|} \le \left( \frac{e \cdot c(m/2^i) \ln n}{|N|} \right)^{|N|}
    = \left(\frac{e \cdot c(m/2^i) \ln n}{c' (m/n) \ln n} \right)^{|N|}
    = \left(\frac{e \cdot c \cdot n}{c' 2^i} \right)^{|N|}
\]
ways of selecting $N$ out of $M$. Thus, the probability that there exists a subset $N$ of $M$ of size at least $c' (m/n) \ln n$ that is part of $S_{i,j}$ is at most
\[ 
    \left(\frac{e \cdot c \cdot n}{c' 2^i} \right)^{|N|} \cdot  \left( \frac{2^i}{n} \right)^{|N|}
    = \left( \frac{e \cdot c}{c'} \right)^{|N|} 
\]
which is polynomially small if $c' \ge 2e \cdot c$.
\end{proof}

We conclude the analysis of the honest setting by bounding the work the supervisor, the source and the target perform.

\begin{lemma}
In the honest setting, the source and the target perform $\bigO(m)$ total work during the execution of the supervised mergesort algorithm.
\end{lemma}
\begin{proof}
To send and receive the data items at the beginning and the end of the algorithm, the source and the target perform $\bigO(m)$ work. In addition to that, during the preprocessing, the source shuffles the input ($\bigO(m)$ work), selects $n$ random quantiles $(\bigO(n)$ work) and sorts them $(\bigO(n\log n)$ work). Since $m=\Omega(n\log n)$, all these steps can be performed with $\bigO(m)$ work. Concatenating the results at the end can easily be done work $O(m)$ work by the target as well.
\end{proof}

\subsection{Malicious Setting}
    
In this section, we describe how to adapt the honest version of the supervised mergesort algorithm to be able to handle a $\beta$-fraction of malicious workers. To this end, we will first describe the additional checks the workers, the source, the target, and the supervisor have to perform throughout the execution of the algorithm and describe a particularly problematic attack. Afterwards we analyze how the modified algorithm performs building on the results from the honest case and the previous sections.

\paragraph*{Handling Malicious Activity}

In the preprocessing part, the supervisor signs each data item and attaches its index resulting from the random permutation to it.

Similar to the DAG in Section~\ref{sec:dag}, the task graph for the honest case is extended by initial lists and final lists of length $c \log n$, for some sufficiently large constant $c \ge 1$. The tasks of each initial list distribute the sorting work among them so that each task only executes $\max\{1,\log (m/n) / (c \log n)\}$ levels of the standard sequential mergesort algorithm in order to sort the $m/n$ data items associated with the corresponding initial node. This reduces the computational work in the initial lists to $O((m/n)\log m/\log n)$ for each task. Each final list does nothing else than forwarding the sorting result from one task to the next until it is sent to the target at the end.

To handle malicious activity during the computation of the tasks, the supervisor, the source, the target, and the honest workers perform several checks. For each initial node $v_{0,j}$, the source sends the signed $m/n$ data items with indices $j\cdot\nicefrac{m}{n}+1,\dots,(j+1)\cdot\nicefrac{m}{n}$ and the two quantiles required for the first merging step to $v_{0,j}$. Whenever these data items and quantiles are forwarded from one task to the next in the initial list, the worker responsible for the next task is supposed to check whether exactly $m/n$ data items are received, all data items and quantiles are signed, the data items have the right indices, and the sorting part that the worker responsible for the previous task was supposed to execute has been done. If the worker performing the checks is honest, this makes it impossible for an adversarial worker to cheat. If any cheating is discovered or nothing is sent to the worker, it sends a REJECT to the supervisor. Otherwise, we distinguish between two cases. If the honest worker is not executing the last task in the initial list, it simply sends DONE to the supervisor. If it is executing the last task, it tells the supervisor how many of the $m/n$ data items it intends to send to its two successors in the task graph. The supervisor will check if these two numbers add up to $m/n$ and otherwise rejects the answer, which causes it to select a new worker for the last task. This certainly also happens if the worker does not reply to the supervisor.

Next, suppose that the supervisor wants a task $v$ in the merging part to be executed that requires data items from tasks $u_1$ and $u_2$. Let $p(v)$ be the worker selected for $v$ and $p(u_i)$ be the worker last selected for $u_i$, $i\in \{1,2\}$. We assume that $p(u_i)$ already informed the supervisor how it will split up its set of data items so that the supervisor already knows how many data items should be sent from each $p(u_i)$ to $v$. The supervisor will then inform $p(v)$ about these numbers, the two quantiles relevant for the range of $v$, and the quantile needed for the splitting of the sorted set. $p(v)$ is supposed to check whether it gets the right number of data items from each $p(u_i)$ and the data items are sorted, all data items are signed, the data items belong to its segment (which can be checked via their indices), and the data items belong to the quantile range provided by the supervisor. Analogously to the DAG case, for every $p(u_i)$ for which $p(v)$ discovers any cheating or does not receive anything, $p(v)$ includes $p(u_i)$ in its REJECT($R$) message to the supervisor. Otherwise, it merges the two data sets and informs the supervisor about the number of the data items it intends to send to its two successors in the task graph. The supervisor will check if these two number add up to the total number of data items received by $p(v)$ and otherwise rejects the answer, which causes it to select a new worker for $v$. This certainly also happens if $p(v)$ does not reply. Unfortunately, there is an opportunity for an adversarial worker to cheat here because it can propose to send a lower number along one direction and a higher number along the other direction so that the sum is still correct. An honest worker receiving this lower number won't notice that some data items are missing. However, if the worker that is supposed to receive the higher number is honest, it will discover the cheating since the adversarial worker cannot make up additional data items as they must be signed. 

Since the final list simply forwards the data items, the only issue that needs to be checked here is whether the same number of data items is forwarded from one task to the next and the data items are signed and sorted.

Whenever the target receives data items from a final node, it performs the same checks a worker receiving the data items would perform, potentially resulting in the selection of a new worker for that node. Note that the target will certainly discover any cheating since the total number of data items has to add up to $m$, which implies that adversarial workers cannot propose numbers of data items deviating from the true value without the target noticing that at some point.

\paragraph*{Analysis}

Since the delay sequence argument in Lemma~\ref{th:dag-runtime} also applies to the supervised mergesort algorithm so that we get the following result (choose $\epsilon>0$ so that $(\nicefrac{1}{2(2d+1)})^{2+\epsilon}=1/200$).
    
\begin{corollary}\label{lem:msm:termination_runtime}
If $\beta\leq 1/200$ then for any adversarial strategy, the supervised mergesort algorithm terminates within $\bigO(\log n)$ rounds, w.h.p.
\end{corollary}
\begin{proof}
Even though cheating might not be noticed in the special case that the adversary proposes wrong numbers for the number of data items it intends to send to its successors, the delay sequence argument is still valid, which is simply based on the fact that there must be a reason why the supervisor asks a worker to execute a task at a specific time: Either because a preceding task was successfully executed at the previous round, or the worker executing that task in the previous round didn't reply or sent the wrong information to the supervisor, causing the supervisor to select a new worker for that task, or a successor of that task caused a rollback to that task. Thus, the only concern successful cheating might raise is that the output at the target is incorrect. However, this is not an issue for the runtime bound.

\end{proof}

Since the target only accepts a correct output by construction, the correctness of the supervised mergesort algorithm is immediate. Next we analyze the verification work of the workers. Recall that a worker needs to check whether it gets the right number of data items from each predecessor and the data items are sorted (resp. in the initial list case, that the mergesort algorithm has been applied to these up to a certain point), all data items are signed, the data items belong to its segment (which can be checked via their indices), and the data items belong to the quantile range provided by the supervisor. Obviously, this can be done in constant time per data item (given that a signature can be checked in constant time) so that the work required to verify data items does not exceed the work required to send and receive them. Moreover, the tasks of the sorting part are guaranteed to just deal with a constant number of layers of the mergesort algorithm for $m/n$ data items, and the tasks of the merging part have to deal with at most $O((m/n) \log n)$ data items, w.h.p., according to Lemma~\ref{lem:mergework}. Thus, we get:

\begin{corollary}
In the malicious setting, every task requires at most $\bigO((m/n)(\log m/\log n +\log n))$ work, w.h.p.
\end{corollary}
    
However, using such a work bound would result in a total work that is a $\log n$ factor higher than the work in the honest case. Thus, we will do a more refined analysis of the number of subsequences of data items of a certain size that do not contain a quantile. Let a subsequence be any consecutive sequence of data items from the sorted sequence $x_1<x_2< \ldots <x_m$.

\begin{lemma}\label{lem:msm:work_distribution_bound}
There are at most $\frac{n}{(s/2)e^{(s/2)-2}}+c\ln n$ disjoint subsequences of size at least $s\cdot\nicefrac{m}{n}$ that contain no quantiles, w.h.p, for $s\geq 1$.
\end{lemma}
\begin{proof}
We start by reducing the number of subsequences we have to consider. Certainly, for every subsequence $S$ of size $s\cdot\nicefrac{m}{n}$ there must exist a subsequence $S' \subseteq S$ of size $s\cdot\nicefrac{m}{2n}$ that starts with some $x_i$ where $i$ is a multiple of $s\cdot\nicefrac{m}{2n}$. Therefore, 
\[
    \Pr[\text{$S$ contains no quantiles}] \leq \Pr[\text{$S'$ contains no quantiles}]
\]
Thus, we can bound the probability that $k$ disjoint subsequences of size at least $s\cdot\nicefrac{m}{n}$ contain no quantiles by considering $k$ disjoint subsequences $S_1,\dots, S_k$ of size $s\cdot\nicefrac{m}{2n}$ that start with some $x_i$ where $i$ is a multiple of $s\cdot\nicefrac{m}{2n}$.
\begin{align*}
    \Pr[\text{$S_1,\dots,S_k$ contain no quantiles}]
    &\leq \frac{{m-k\cdot \nicefrac{sm}{2n}\choose n}}{{m \choose n}} \\
    &\leq e^{-\frac{k \cdot \nicefrac{sm}{2n}}{m}n}\tag{Proof of Lemma~\ref{lem:mergework}}\\
    &= e^{-k\cdot\nicefrac{s}{2}}
\end{align*}
Since there are $m/(s\cdot\nicefrac{m}{2n})=\nicefrac{2n}{s}$ many starting points for such subsequences, there are ${2n/s \choose k}$ ways of picking $k$ starting points so that $S_1,\ldots,S_k$ are disjoint. Hence, we can use the union bound to obtain for $k=\frac{n}{(s/2)e^{\nicefrac{s}{2}-2}} + c\ln n$:
\begin{align*}
    \Pr[\text{$S_1,\dots,S_k$ contain no quantiles}] &\leq 2 \cdot {2n/s \choose k} \cdot e^{-k\cdot\nicefrac{s}{2}}\\
    &\leq \left(\frac{2en}{s k}\right)^k \cdot e^{-k\cdot\nicefrac{s}{2}}
    = \left(\frac{n}{(s/2)e^{\nicefrac{s}{2}-1}k}\right)^k\\
    &= \left(\frac{1}{e}\right)^{(c+1)\ln n} \leq n^{-c}
\end{align*}
\end{proof}

We continue by bounding the total work performed by the workers.

\begin{lemma}\label{lem:msm:peerwork}
The total computational work performed by the honest workers is $O(m\log m)$ and the total communication work $O(m \log n)$ on expectation.
\end{lemma}
\begin{proof}
In the same way that Lemma~\ref{th:line-middle} extends the expected constant work bound for the source to honest workers in a path graph, Lemma~\ref{lem:advsupervisor} can be extended to show that every sorting task will, on expectation, be executed by an honest worker just $O(1)$ times. Thus, the total expected work resulting from the sorting tasks is
\[
  O(n \log n \cdot (m/n)\log m/\log n) = O(m \log m) . 
\]
Moreover, since every sorting task has to deal with exactly $m/n$ data items, the total communication work for these is
\[
  O(n \log n \cdot (m/n)) = O(m \log n) .
\]
It remains to bound the total work of the other tasks. Due to Corollary~\ref{lem:msm:termination_runtime} at most $O(\log n)$ rounds are needed by the supervised mergesort algorithm, w.h.p. Thus, every task is executed at most $O(\log n)$-times. Moreover, because at most $n$ tasks are executed in each round, the total number of task executions is $O(n \log n)$, w.h.p. The worst case for the work bound is reached if the $(n \log n)/\log n$ most expensive tasks are all executed $O(\log n)$ times and the rest just once. Similar to Lemma~\ref{lem:mergework}, Lemma~\ref{lem:msm:work_distribution_bound} can be extended to show that for some $s=\Theta(\log \log n)$ there can be at most $n/\log n$ nodes at each level of the task graph with at least $s \cdot m/n$ many data items. As a worst case, assume that there are exactly $n/\log n$ such nodes in each level, and each of them has $O((m/n) \log n)$ data items. Then the total work required for these is $O(\log n \cdot (n/\log n) \cdot (m/n) \log n) = O(m \log n)$. For the rest of the nodes, there are at most $O(\frac{n}{s e^{s/2}})$ many per level with a work of at least $s \cdot m/n$. Summing up all these work bounds results a total work of
\[
  \sum_{s \ge 1} \log n \cdot O \left( \frac{n}{s e^{s/2}} \right) \cdot s \cdot m/n
  = O(m \log n) \sum_{s \ge 1} \frac{1}{e^{s/2}} = O(m \log n)
\]
Together with the work bound for the expensive tasks this results in a total expected computational work of $O(m \log n)$. Since the computational work is an upper bound for the communication work, the lemma follows.
\end{proof}

Certainly, an upper bound on the computational work is also an upper bound on the communication work of the workers. Next, we bound the verification work of the supervisor.
    
\begin{lemma}
The supervisor performs $O(n\log n)$ work in total for the verification of the outputs of tasks, w.h.p.
\end{lemma}
\begin{proof}
The supervisor assigns $O(n\log n)$ workers to nodes and introduces each to a constant number of other workers (or the source and target). In addition, it receives $O(n)$ quantiles from the source and sends $O(1)$ to each worker assigned to a task. To verify the output of a single task, the supervisor has to receive the amount of data the worker assigned to the task wants to send to each of its successors, it has to compare this amount with the amount the worker received, and it has to notify its successors of the amounts that were announced to be sent. In total, this requires $O(1)$ work per task execution. As there are $O(n \log n)$ task executions in total, the lemma follows.
\end{proof}

Furthermore, Lemma~\ref{th:dag_sendingwork} implies the following result.
    
\begin{corollary}
If $\beta \le 1/12$ then the source performs $O(m)$ work for sending data items in expectation.
\end{corollary}

Finally, Lemma~\ref{th:dag_receivingwork} implies the following result.

\begin{corollary}
If $\beta \le 1/3$ then the target performs $O(m)$ work for receiving data items in expectation.
\end{corollary}

Summing up all lemmas and corollaries, we obtain \Cref{th:merge}.

\section{Conclusion}

We presented a new framework in distributed computing that employs a reliable supervisor, a reliable source, and a reliable target to orchestrate a community of potentially unreliable workers in order to solve problems whose computations can be decomposed into a task graph. Subsequently, we demonstrated the power of this framework by applying it to two important problems: matrix multiplication and sorting.

We believe that our framework can be applied to many other problems that can be parallelized well. The challenging aspect here is that lightweight verification mechanisms have to be found, though this might open up an interesting new direction of distributed verifiability of problems. Furthermore, there is no reason why it shouldn't be possible to apply our framework to the case where a majority of peers is adversarial (by using different scheduling approaches, for example), which would further underline the power of our approach.

\bibliographystyle{plainurl}
\bibliography{main}

\end{document}